\documentclass{article}
\usepackage{arxiv}
\usepackage[utf8]{inputenc} 
\usepackage[T1]{fontenc}    
\usepackage{hyperref}       
\usepackage{url}            
\usepackage{booktabs}       
\usepackage{amsfonts}       
\usepackage{nicefrac}       
\usepackage{microtype}      
\usepackage{subfiles}
\usepackage{doi}

\usepackage{amsmath,amssymb,bm,amsthm}
\usepackage{multirow,makecell} 
\usepackage[english]{babel}
\usepackage{natbib}
\usepackage{setspace}
\usepackage{graphicx,multicol,float}
\usepackage{xcolor}
\usepackage[linesnumbered,ruled,vlined]{algorithm2e}
\usepackage{subcaption}

\newtheorem*{theorem*}{Theorem}
\newtheorem{proposition}{Proposition}

\newtheorem*{lemma*}{Lemma}
\newtheorem*{corollary*}{Corollary}
\newcommand{\norm}[1]{\left\lVert#1\right\rVert}

\title{Outlier Detection for Multi-Network Data}

\date{June 24, 2022}	

\author{ Pritam Dey \\
	Department of Statistical Science\\
	Duke University\\
	Durham, NC 27708, USA \\
	\texttt{pritam.dey@duke.edu} \\
	\And
	Zhengwu Zhang\\
	Statistics and Operations Research\\
	University of North Carolina at Chapel Hill\\
	Chapel Hill, NC 27599, USA \\
	\texttt{zhengwu\_zhang@unc.edu} \\
    \And
	David B. Dunson\\
	Department of Statistical Science\\
	Duke University\\
	Durham, NC 27708, USA \\
	\texttt{dunson@duke.edu} \\
}

\renewcommand{\shorttitle}{Outlier Detection for Multi-Network Data}

\begin{document}
\maketitle
\graphicspath{{figures/}}

\begin{abstract}
It has become routine in neuroscience studies to measure brain networks for different individuals using neuroimaging.  These networks are typically expressed as adjacency matrices, with each cell containing a summary of connectivity between a pair of brain regions.  There is an emerging statistical literature describing methods for the analysis of such multi-network data in which nodes are common across networks but the edges vary.  However, there has been essentially no consideration of the important problem of outlier detection. In particular, for certain subjects, the neuroimaging data are so poor quality that the network cannot be reliably reconstructed.  For such subjects, the resulting adjacency matrix may be mostly zero or exhibit a bizarre pattern not consistent with a functioning brain.  These outlying networks may serve as influential points, contaminating subsequent  statistical analyses.  We propose a simple Outlier DetectIon for Networks (ODIN) method relying on an influence measure under a hierarchical generalized linear model for the adjacency matrices.  An efficient computational algorithm is described, and ODIN is illustrated through simulations and an application to data from the UK Biobank. ODIN was successful in identifying moderate to extreme  outliers. Removing such outliers can significantly change inferences in downstream applications.
\\
\textbf{Code availability:} ODIN has been implemented in both Python and R and these implementations along with other code are publicly available at \href{http://www.github.com/pritamdey/ODIN-python}{github.com/pritamdey/ODIN-python} and \href{http://www.github.com/pritamdey/ODIN-r}{github.com/pritamdey/ODIN-r} respectively.\\
\end{abstract}


\section{Introduction}

In recent years there has been substantial interest in brain functional and structural connectomics.  Although the proposed methodology is general, we are motivated by structural connectomes, consisting of the collection of white matter fiber bundles connecting various regions of the brain. Recent advances in non-invasive brain imaging have made available large brain imaging datasets including the Human Connectome Project \citep{van2013wu}, the Adolescent Brain Cognitive Development Study \citep{casey2018adolescent} and the UK Biobank \citep{miller2016multimodal}. Relying on the availability of such datasets, there is a large literature on statistical analysis of brain networks developing models for characterizing inter-individual variation \citep{wang2019common, durante2017nonparametric, aliverti2019spatial}, ANOVA-like hypothesis testing \citep{ginestet2017hypothesis}, and network regression \citep{wang2017bayesian, zhang2018network, zhang2019tensor}.

An important issue to consider in brain network analysis is reconstruction error from available neuroimaging data.  Even if image acquisition is conducted correctly for a subject who remains still in the scanner, there is inevitably some amount of error in conducting the inverse problem of inferring the white matter fiber bundles based on indirect measurements \citep{fornito2013graph}. Moreover, due to the long preprocessing pipeline in structural connectome reconstruction, a small measurement error in the raw imaging data can be amplified and yield substantial errors in the inferred structural connectome. In this article, instead of being concerned with small measurement errors that are difficult to distinguish from actual biological variation, we focus on identifying outlying brain networks that are almost certainly attributable to measurement errors in reconstructing the connectome.  Such gross errors can potentially arise due to problems during the data collection phase; for example, due to non-negligible movement of the patient in the scanner \citep{baum2018impact} or mistakes in preprocessing large of amounts of data using complex structural connectome reconstruction pipelines \citep{zhang2018mapping}. 

Some examples of outlying brain connectomes are shown in Figure \ref{adjfigintro}. A major concern is that such outlying networks may serve as influential observations in statistical analyses of brain connectomes, leading to degradation of the results.  For example, suppose we are attempting to represent variation across subjects in their brain connectomes via an embedding, such as the PCA method of \citet{zhang2019tensor}.  Including outlying brains can lead to poor quality embeddings, as PCA needs to characterize not just normal biological variation in the brain connectomes but also the outlying brains.  Analyses seeking to infer relationships between brain networks and human traits can similarly be contaminated by outliers, obscuring true relationships.  Current practice tends to either ignore outliers or to apply informal quality controls, such as visual examination of connectome matrices \citep{alfaro2018image}.  Such checks are overly-subjective and time consuming for large datasets.  Potentially, one can instead apply existing statistical outlier detection methods \citep{hawkins1980identification} to low-dimensional summary statistics of the connectome (for example, graph metrics such as node degree, average path length, and eigenvector centrality).  However, such an approach will be highly sensitive to the summaries chosen, and may miss certain outlying networks or remove non-outlying networks.

\begin{figure*}
\begin{subfigure}{.19\linewidth}
    \includegraphics[width=\linewidth]{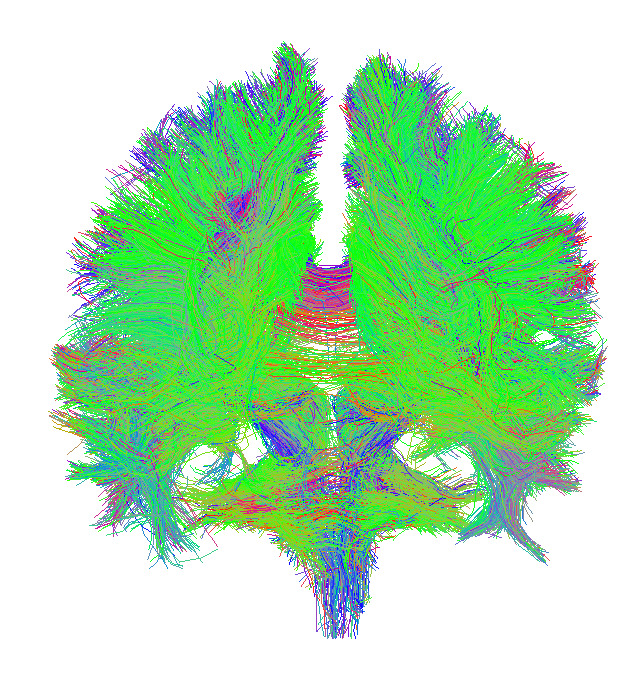}
\end{subfigure}
\hfill
\begin{subfigure}{.19\linewidth}
    \includegraphics[width=\linewidth]{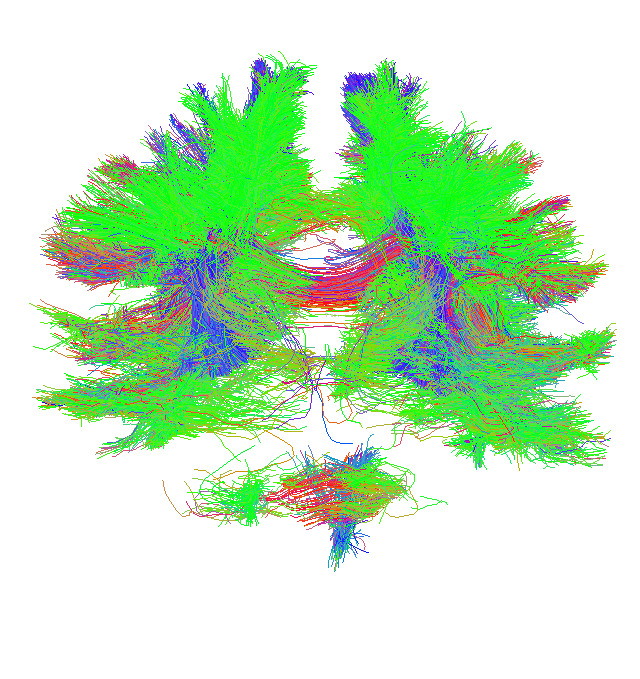}
\end{subfigure}
\hfill
\begin{subfigure}{.19\linewidth}
    \includegraphics[width=\linewidth]{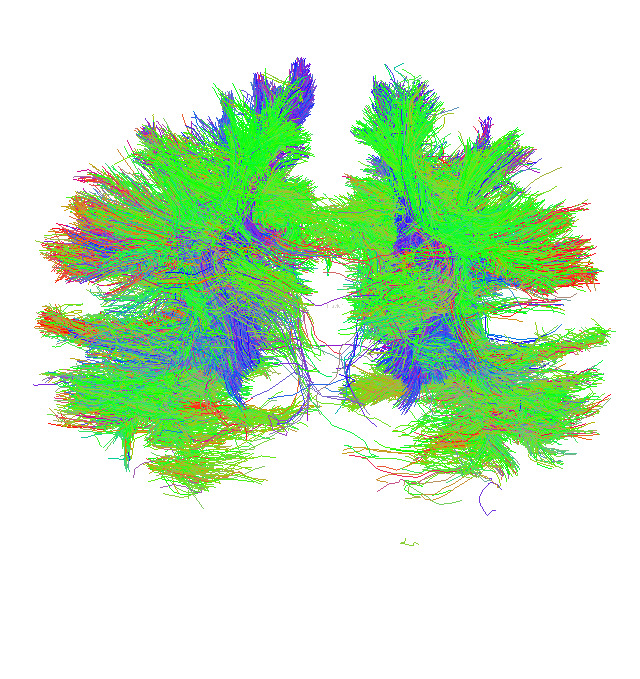}
\end{subfigure}
\hfill
\begin{subfigure}{.19\linewidth}
    \includegraphics[width=\linewidth]{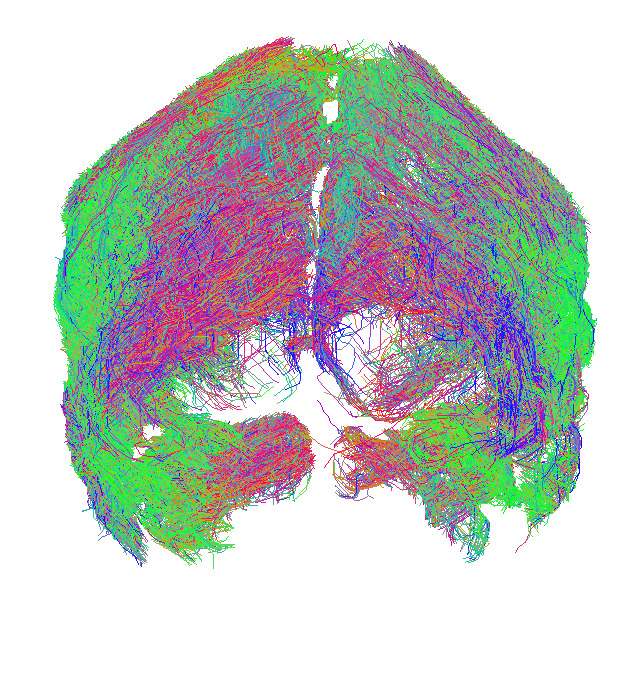}
\end{subfigure}
\hfill
\begin{subfigure}{.19\linewidth}
    \includegraphics[width=\linewidth]{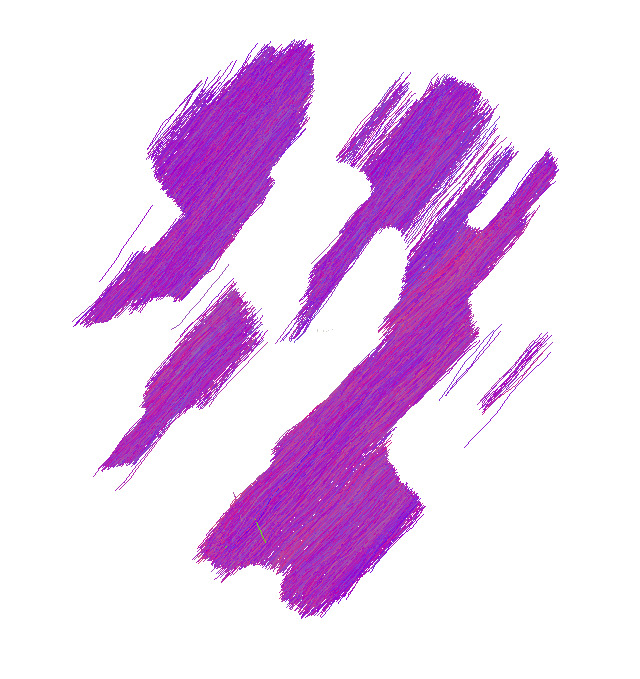}
\end{subfigure}
\\
\begin{subfigure}{.19\linewidth}
    \includegraphics[width=\linewidth]{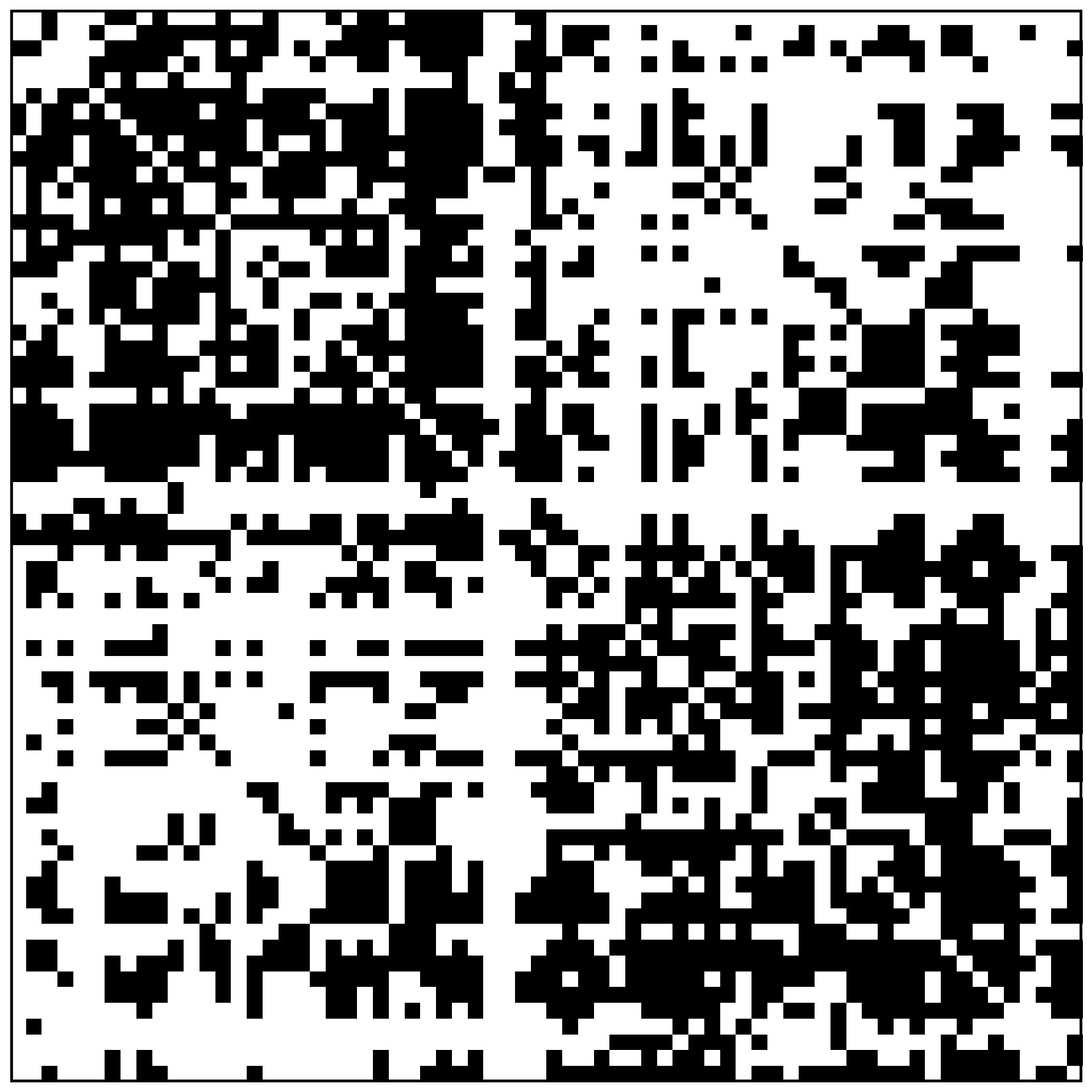}
\caption{}
\end{subfigure}
\hfill
\begin{subfigure}{.19\linewidth}
    \includegraphics[width=\linewidth]{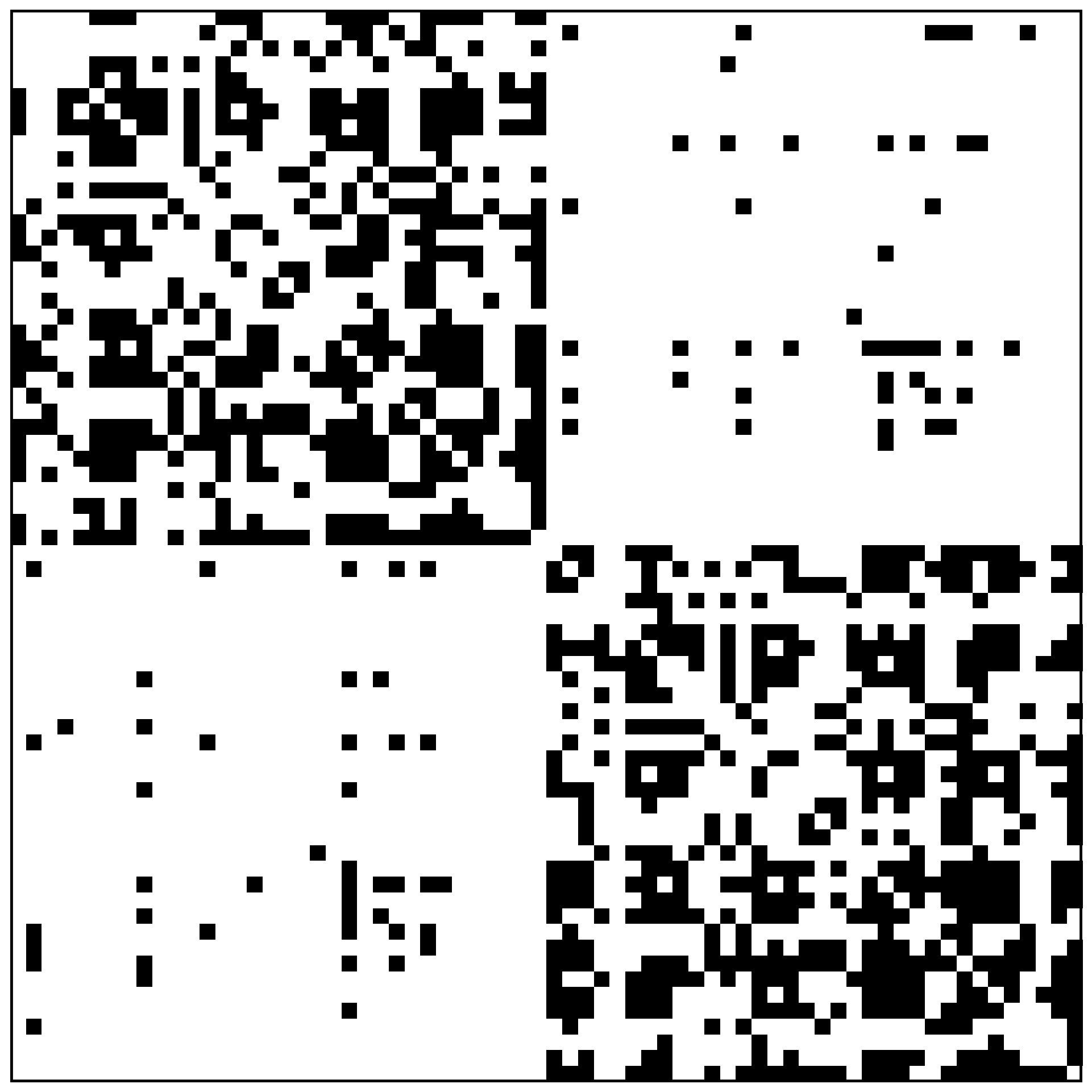}
\caption{}
\end{subfigure}
\hfill
\begin{subfigure}{.19\linewidth}
    \includegraphics[width=\linewidth]{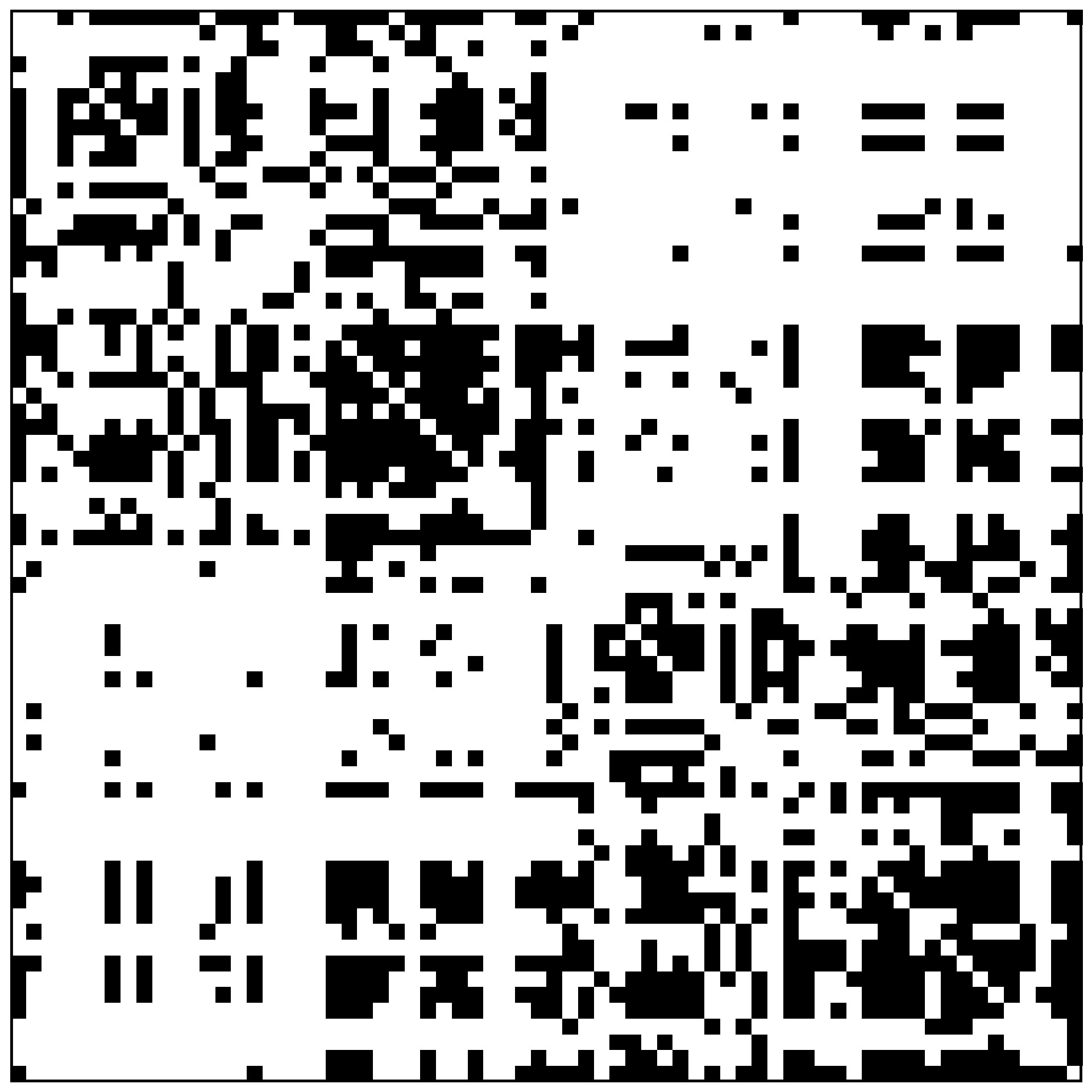}
    \caption{}
\end{subfigure}
\hfill
\begin{subfigure}{.19\linewidth}
    \includegraphics[width=\linewidth]{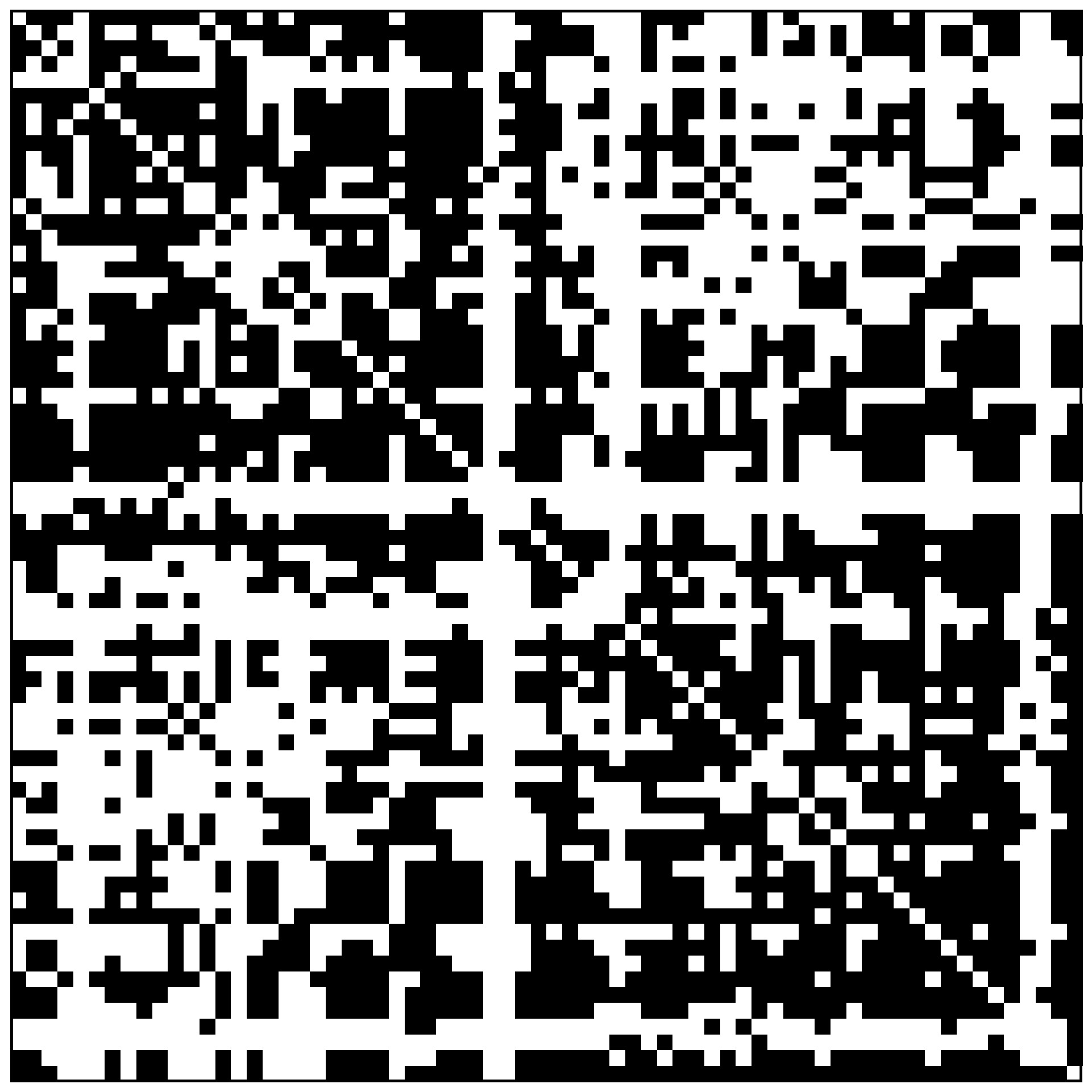}
    \caption{}
\end{subfigure}
\hfill
\begin{subfigure}{.19\linewidth}
    \includegraphics[width=\linewidth]{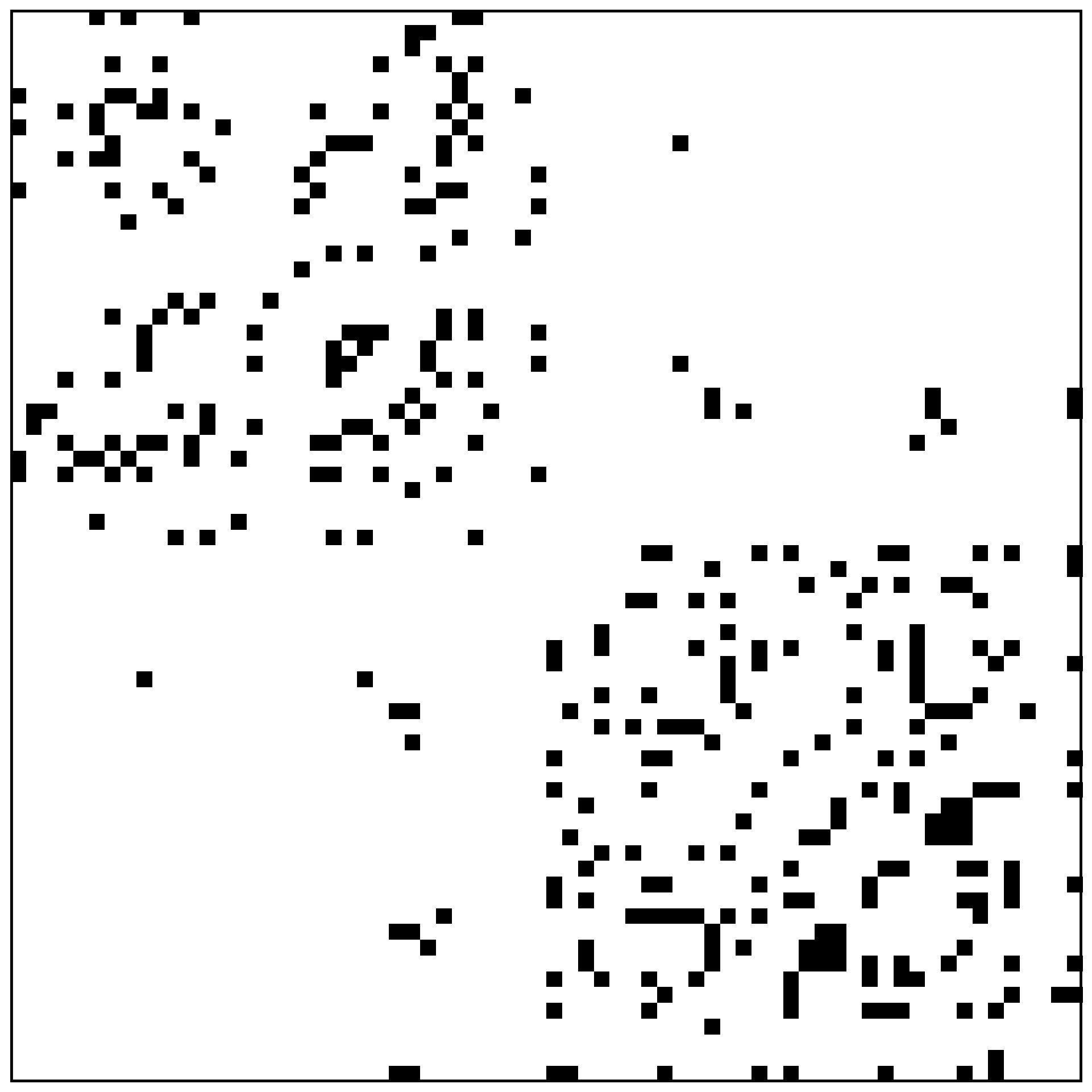}
    \caption{}
\end{subfigure}
\caption{Brain fiber streamlines from diffusion MR imaging (top row) with corresponding binary adjacency matrices (bottom row) of some subjects from the UK Biobank dataset. The streamlines are visualized using TrackVis \citep{wang2007trackvis} and colored by orientation (i.e., left to right: red, anterior to posterior: green, superior to inferior: blue). In these tractography diagrams, the anterior side of the brain is facing inwards into the page. The matrices in the bottom row are the ones used by ODIN. These adjacency matrices are not directly available from UK Biobank. We preprocessed the raw data using the PSC pipeline \citep{zhang2018mapping} to extract these adjacency matrices. In these matrices, black indicates presence of at least one fiber connecting the two corresponding regions and white represents absence of such fibers. The brain network represented by the streamlines and adjacency matrix in (a) is a typical non-outlier. The networks shown in (b) - (e) are outliers of various kinds selected from among the outliers detected by ODIN. }
\label{adjfigintro}
\end{figure*}

Ideally, we could apply an outlier detection method specifically designed for multi-network data, to identify individual networks that are fundamentally different than the bulk of the networks in a dataset.  However, to our knowledge, there are no such methods available in the literature.  To bridge this gap, we propose a simple model-based
Outlier DetectIon for Networks (ODIN) 
method relying on a hierarchical logistic regression model and a measure of the statistical influence of each subject's brain connectome on the parameter estimates in this model. The logistic model includes prior knowledge of the anatomical structure of the brain. Although our initial approach is for binary adjacency matrices conveying information about connectivity/no-connectivity of pairs of ROIs, ODIN can be trivially extended to weighted adjacency matrices by using alternative generalized linear models (GLMs) in place of logistic regression.  The model can also include covariate information about the subject, such as age and gender. A key advantage of ODIN is computational scalability to massive datasets containing tens of thousands of connectomes.

We apply ODIN to 18,083 structural connectomes extracted from a large cohort, the UK Biobank data, which has a number of moderate to extreme outliers. Figure \ref{adjfigintro} shows a non-outlier and four different outliers from this dataset. Our objective is to detect outliers, and not to use the logistic model for inference directly.  ODIN can be used in a data cleaning step to remove outliers, or can provide numerical influence scores that can be used for down-weighting of overly influential observations in subsequent robust analyses.  We will demonstrate the former approach in simulation studies and an application involving the UK Biobank data (henceforth refered to as UKB data).

\section{Method}

It is common practice in the brain network analysis literature to partition the brain into small regions based on anatomical considerations \citep{desikan2006automated}. These small regions are known as regions of interest (ROIs). The brain structural connectome is thus summarized as a network with these ROIs as nodes and the fibers connecting these ROIs are the edges in the network. We use this framework for our outlier detection method, ODIN. Let $A_i$ denote the binary adjacency matrix of an undirected brain network on $V$ ROIs (assuming no self connections) for subject $i$, $i=1,2,\hdots,N$. The symmetry of the $A_i$'s coupled with the lack of self connections allows us to denote our networks in terms of the vectors $\bm{a}_i = (A_{i[2,1]},A_{i[3,1]},A_{i[3,2]},\hdots,A_{i[V,V-1]})^T=(a_{i1},a_{i2},\hdots,a_{iL})^T$. Here $L=\frac{V(V-1)}{2}$. Each $l \in \{1,2,...,L\}$ represents an edge in the network connecting two ROIs, say $u$ and $v$. For each subject $i$ and edge $l$ representing the pair of ROIs $\{u,v\}$, we let $hemi(u)$ and $lobe(u)$ represent the hemisphere and the lobe location of the ROI $u$, respectively. Similarly, $hemi(v)$ and $lobe(v)$ does the same for $v$. We model: 
\begin{align}
    {a}_{il} & \sim Bernoulli(\pi_{il}) \nonumber \\ 
    logit(\pi_{il}) &= z_l + \beta_{i,hemi(u),hemi(v)}+\beta_{i,lobe(u),lobe(v)}.
\label{eq:logit}
\end{align}

Model \eqref{eq:logit} characterizes variation across subjects in their brain networks in a semi-parametric manner. The baseline model for the connection probabilities,  $\pi_{0l}=1/(1+e^{-z_l})$, 
is nonparametric in allowing fully flexible variation in these probabilities across different pairs of ROIs.  This avoids imposing a particular model on the population-averaged connection probabilities between different regions of the brain.

We then allow a more restricted type of variation across subjects, with subject-specific deviations from the baseline edge probabilities depending on hemisphere and lobe locations of the two brain ROIs forming the edge. The coefficient $\beta_{i,hemi(u),hemi(v)}$ allows connection probabilities between regions $u$ and $v$ in subject $i$'s brain to vary depending on which hemisphere $u$ and $v$ belong to. Similarly the coefficient $\beta_{i,lobe(u),lobe(v)}$ allows the connection probability to vary depending on the lobe locations of $u$ and $v$. Since there are only two hemispheres and the number of lobes is relatively small, the number of subject-specific parameters for each subject is small. Due to the symmetric nature of the networks, we may assume: $\beta_{i,h_1,h_2}=\beta_{i,h_2,h_1}$ and $\beta_{i,l_1,l_2}=\beta_{i,l_2,l_1}$ for every choice of hemispheres $h_1$ and $h_2$ and lobes $l_1$ and $l_2$. We collect all these lobe and hemisphere parameters for subject $i$ into a vector $\beta_i$. These parameters $\beta_i$ allow for variations across subjects, in terms of lower or higher numbers of connections between specific lobes and hemispheres.  

Using the notation we introduced above, we can write $logit(\pi_{il})=z_l + x_l^T\beta_i$ for a suitable vector $x_l$ of 0's and 1's. We stack these vectors $x_l$ into a matrix $X=(x_1,x_2,...x_L)^T$. Further, we represent the vector of parameters $(z_1,z_2,...,z_L)^T$ as $Z$ and the vector $(\pi_{i1},\pi_{i2},...,\pi_{iL})^T$ as $\bm{\pi}_i$.

Due to the large number of parameters in the baseline component and to allow cases in which certain edges are connected (or disconnected) for all subjects in the sample, we add a ridge (or $l_2$) regularization to $z_l$. Estimation of the parameters $Z$ and the $\beta_i$'s then proceeds by maximizing the following penalized log-likelihood:
\begin{equation}
    \label{loglikelihood}
    \mathcal{L}(Z,\{\beta_j\}_{j=1}^N)=\frac{1}{N}\sum_{i=1}^N \sum_{l=1}^L \left[a_{il}\eta_{il}-log(1+e^{\eta_{il}})\right] -\frac{\lambda}{2}\norm{Z}^2_2,
\end{equation}

where $\eta_{il}=logit(\pi_{il})=z_l + x_l^T\beta_i$.

The Hessian matrix of the objective function has a block arrowhead structure, which can be exploited to obtain a fast algorithm for estimation. While the estimation step can be done using a variety of different algorithms, we recommend a slightly modified version of the standard Majorization - Minimization (MM) algorithm for logistic regression to find the parameter values that maximize the equation \eqref{loglikelihood} (technically we minimize its negative). The algorithm is discussed in detail in section S1.1 in the supplementary material. This algorithm is chosen as it scales well with sample size, $N$ and the number of edges, $L$. In fact each iteration of the algorithm has linear time complexity in $N$ and $L$. This is verified in the simulation section.

In exploring methodology for outlier detection of brain networks, we considered a wide variety of complex and flexible models for characterizing variation across individuals in their brain structure including a hierarchical latent space model. However, computational time was a substantial barrier to implementation of more elaborate models. Our objective being the identification of outlying networks as a step in the cleaning and preprocessing of the data before subsequent statistical analyses, we wanted our model to be simple and computationally efficient while still being effective in detecting outliers. Our relatively simple logistic model did an excellent job in fulfilling all of these objectives.  Indeed, if a model is too flexible in characterizing inter-individual variability, then it can end up not flagging any of the networks as outliers. To identify potentially outlying subjects through our model, we consider appropriate influence measures.

One way of measuring the influence of subject $i$ on the penalized maximum likelihood estimates (pMLEs) for the parameters in equation \eqref{loglikelihood} is by the change in the pMLE of $Z$ when subject $i$ is dropped from the sample. More concretely, if the pMLE of $Z$ is denoted by $\widehat{Z}$ and the pMLE of $Z$ is recalculated after dropping subject $i$ from the sample is denoted by $\widehat{Z}_{-i}$ then an influence measure may be $\norm{\widehat{Z}_{-i}-\widehat{Z}}$ which is analogous to DFBETA, a popular influence measure in regression analysis. However, since the maximizer of equation \eqref{loglikelihood} does not have an analytical form and has to be numerically calculated, it is not computationally feasible to compute the pMLE exactly for the case of omission of each subject when the number of subjects is large. So instead, we suggest for every subject $i$, to do just one step of the Newton-Raphson algorithm for calculation of the pMLE after dropping subject $i$ but with the already calculated pMLE (with all subjects included) as the initial values. We do this since dropping subject $i$ leaves the resultant log-likelihood still close to the original log-likelihood in \eqref{loglikelihood}.  We denote this approximation by $Z^{(1)}_{-i}$. Hence, our influence measure is simply the Euclidean distance between $\widehat{Z}$ and $Z^{(1)}_{-i}$. The following result gives an explicit formula in terms of the overall pMLEs of $Z$ and the $\beta_i$'s. 

\begin{proposition}
\label{mainproposition}
Suppose the pMLE obtained by maximizing (\ref{loglikelihood}) is $(\widehat{Z},\widehat{\beta}_1,...,\widehat{\beta}_N)$. Then the update in the common parameters $Z$ in the first step of the Newton-Raphson algorithm for re-estimation of the pMLE after deletion of the $i^{th}$ unit starting from $(\widehat{Z},\widehat{\beta}_1,...,\widehat{\beta}_{i-1},\widehat{\beta}_{i+1},...\widehat{\beta}_N)$ is
\begin{equation}
    \begin{split}
         Z^{(1)}_{-i} - \widehat{Z} = -\left[\lambda (N-1) I +\sum_{j \neq i}  W_j \right. \left. -\sum_{j \neq i} B_{j}^T Q_j^{-1} B_{j} \right]^{-1}
        \left[(\bm{a}_i - \widehat{\bm{\pi}}_i)-\lambda\widehat{Z}\right]
    \end{split}
\end{equation}
where $W_j$ is a diagonal matrix with entries $\left[\pi_{j1}(1-\pi_{j1}),\right.\allowbreak \left.\pi_{j2}(1-\pi_{j2}),\right.\allowbreak 
\left.\hdots,\right.\allowbreak
\left.\pi_{jL}(1-\pi_{jL})\right]$, $B_j = X^T W_j$, and $Q_j = X^T W_j X$. These are all evaluated at $(\widehat{Z},\widehat{\beta}_1,...,\widehat{\beta}_{i-1},\widehat{\beta}_{i+1},...\widehat{\beta}_N)$.
\end{proposition}
We show a proof of the above proposition in section S1.2 of the supplementary material. The formula for $Z^{(1)}_{-i} - \widehat{Z}$ given in proposition \ref{mainproposition} needs one more approximation to make it computationally feasible. The matrix $\left[\lambda (N-1) I +\sum_{j \neq i}  W_j   -\sum_{j \neq i} B_{j}^T Q_j^{-1} B_{j} \right]$ depends on $i$, and hence we need to calculate and invert this huge matrix for each subject. We can avoid that by approximating it instead by $\Gamma =\gamma_N \left[\lambda N I + \sum_{j=1}^N  W_j -\sum_{j = 1}^N  B_{j}^T Q_j^{-1} B_{j} \right]$, where $\gamma_N = (N-1)/N$. This solves the problem as now we only need to invert $\Gamma$ once. After making these approximations, we obtain the following influence measure for brain networks: 
\begin{equation} \label{influence1}
\begin{split}
    IM_1(i) &= \left\Vert \Gamma^{-1} \left[(\bm{a}_i - \widehat{\bm{\pi}}_i)-\lambda\widehat{Z}\right]\right\Vert_2
\end{split}
\end{equation}
A drawback of $IM_1(i)$ as the sole measure of influence is that it does not take into account the fact that the individual-specific parameters, $\beta_i$, can still be radically different for some networks than others. Those networks may not appear influential to the model, even though they can still be outliers. To flag such networks, we need another influence measure along with $IM_1(i)$. A suitable candidate that is easy to compute is the relative distance of the estimated $\widehat{\beta}_i$ from the bulk of the estimated $\widehat{\beta}_j$'s for the entire sample. With this in mind, we use the  Mahalanobis distance of $\widehat{\beta}_i$ from the mean of the $\widehat{\beta}_j$'s, $\overline{\beta}=\sum_{i=1}^N\widehat{\beta}_i/N$, as our second influence measure:
\begin{equation} \label{influence2}
\begin{split}
    IM_2(i) &= \left\Vert(\widehat{\beta}_i-\overline{\beta})^TS(\widehat{\beta})^{-1}(\widehat{\beta}_i-\overline{\beta})\right\Vert_2
\end{split}
\end{equation}
where  $S(\widehat{\beta})=\sum_{i=1}^N(\widehat{\beta}_i-\overline{\beta})(\widehat{\beta}_i-\overline{\beta})^T/N$. We declare the brain network of subject $i$ to be an outlier if either of $IM_1(i)$ or $IM_2(i)$ is large.

Finally, we need to define thresholds for classifying if the influence measures, $IM_1(i)$ or $IM_2(i)$, are ``large.'' For this, we use a data based scheme in which we plot the quantiles of the influence measure in question. Since most of the subjects have small influence measures, the graph remains flat for most of the lower quantiles and then starts to sharply increase as we get to the higher quantiles. The `elbow' of this plot, i.e. the point where the graph starts increasing sharply, can be used as a threshold. We use the kneedle algorithm \citep{5961514} to calculate the elbow point. We describe this thresholding method in more detail in section S2 of the supplement.  

The time complexity for calculation of the influence measures $IM_1(i)$ and $IM_2(i)$ for all networks $i$ is linear in the sample size, $N$, and quadratic in the number of edges, $L$. We verify this empirically in the simulation section. Our simulation studies also demonstrate good separation of outliers and non-outliers in terms of these influence measures.

\section{Simulation Study}

In this section, we carry out a number of simulation experiments to evaluate the performance of our outlier detection algorithm, ODIN, and understand the computational complexity of the associated steps.

\subsection{Computational Complexity}

As summarized in section 2, ODIN has two major steps: (1) Estimation of the model parameters $Z$ and the $\beta_i$'s and (2) Calculating the influence measures $IM_1(i)$ and $IM_2(i)$ for each $i$. In this subsection, we perform simulation experiments to study how sample size, $N$, and the number of edges, $L$, impact computational time for these two steps. 

The first step implements an iterative algorithm, so we investigate the effect of $L$, the number of edges, and $N$, the sample size, on the run-time of each iteration. We conduct two sets of simulations, one by keeping $L$ fixed and varying $N$ and one by keeping $N$ fixed and varying $L$. In these simulations, we generate $Z$ and $\beta_i$'s in our model from the standard Cauchy and Gaussian distributions, respectively, and use these to generate the adjacency matrix from our model. For each of these cases, we run 200 iterations of our algorithm and plot the average run-time for each iteration. As noted in the methods section, each iteration of the model fitting algorithm has linear run-time in both $N$ and $L$.

The next step is  calculation of the influence measures. Again we run two sets of simulations as in the previous paragraph. In each simulation we generate the data as before, iterate our estimation algorithm till convergence and calculate the influence measures. This time the run-time is recorded for only the calculation of the influence measures. Plotting this we conclude that the run-time is linear in the sample size $N$ and quadratic in $L$ as mentioned in the methods section. This quadratic complexity in $L$ does not make ODIN computationally prohibitive for the most popular atlases, as they rely on a moderate but not large number of ROIs.

\begin{figure}
\begin{subfigure}{.24\linewidth}
    \includegraphics[width=\linewidth]{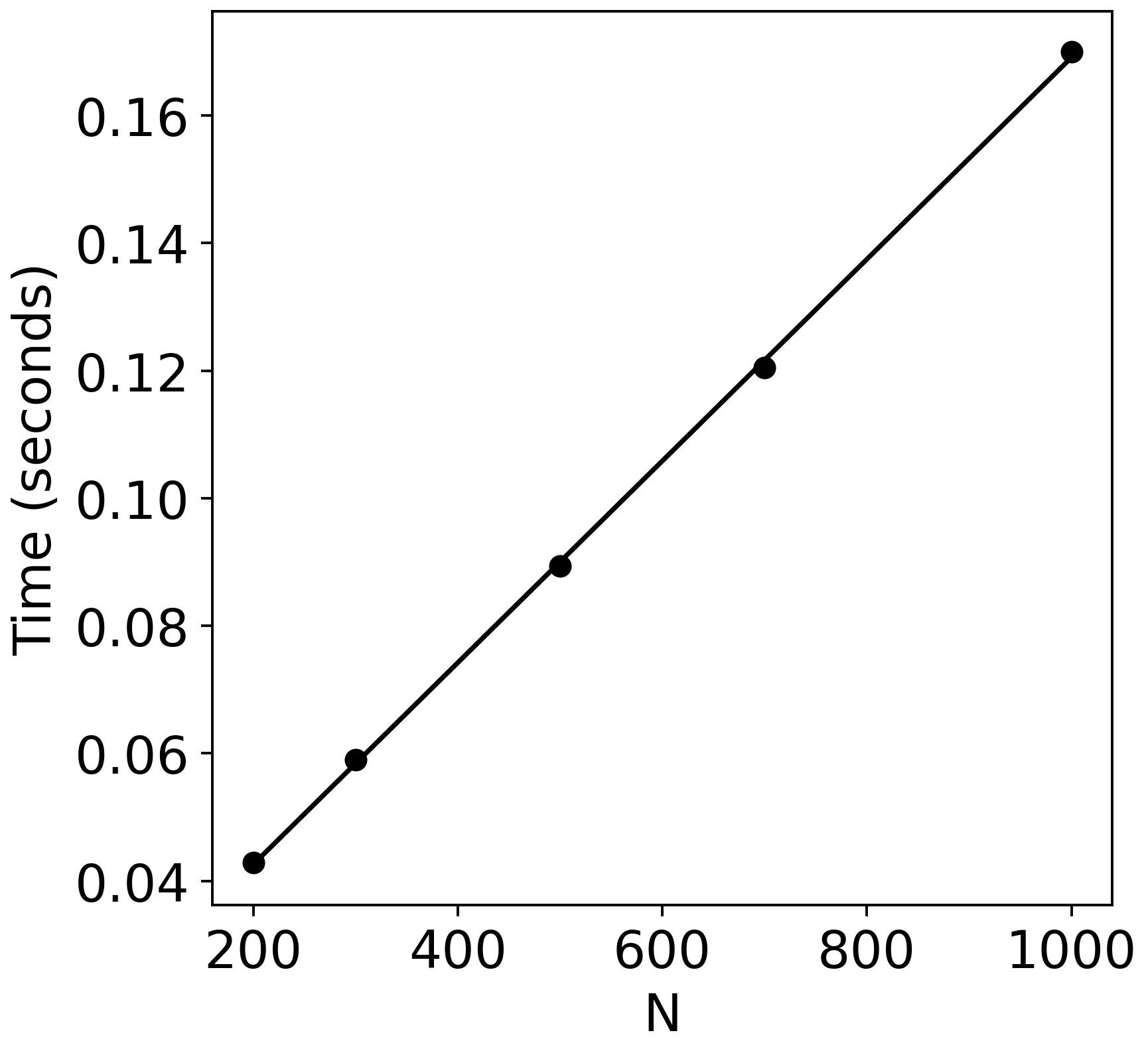}
    \caption{}
\end{subfigure}
\hfill
\begin{subfigure}{.24\linewidth}
    \includegraphics[width=\linewidth]{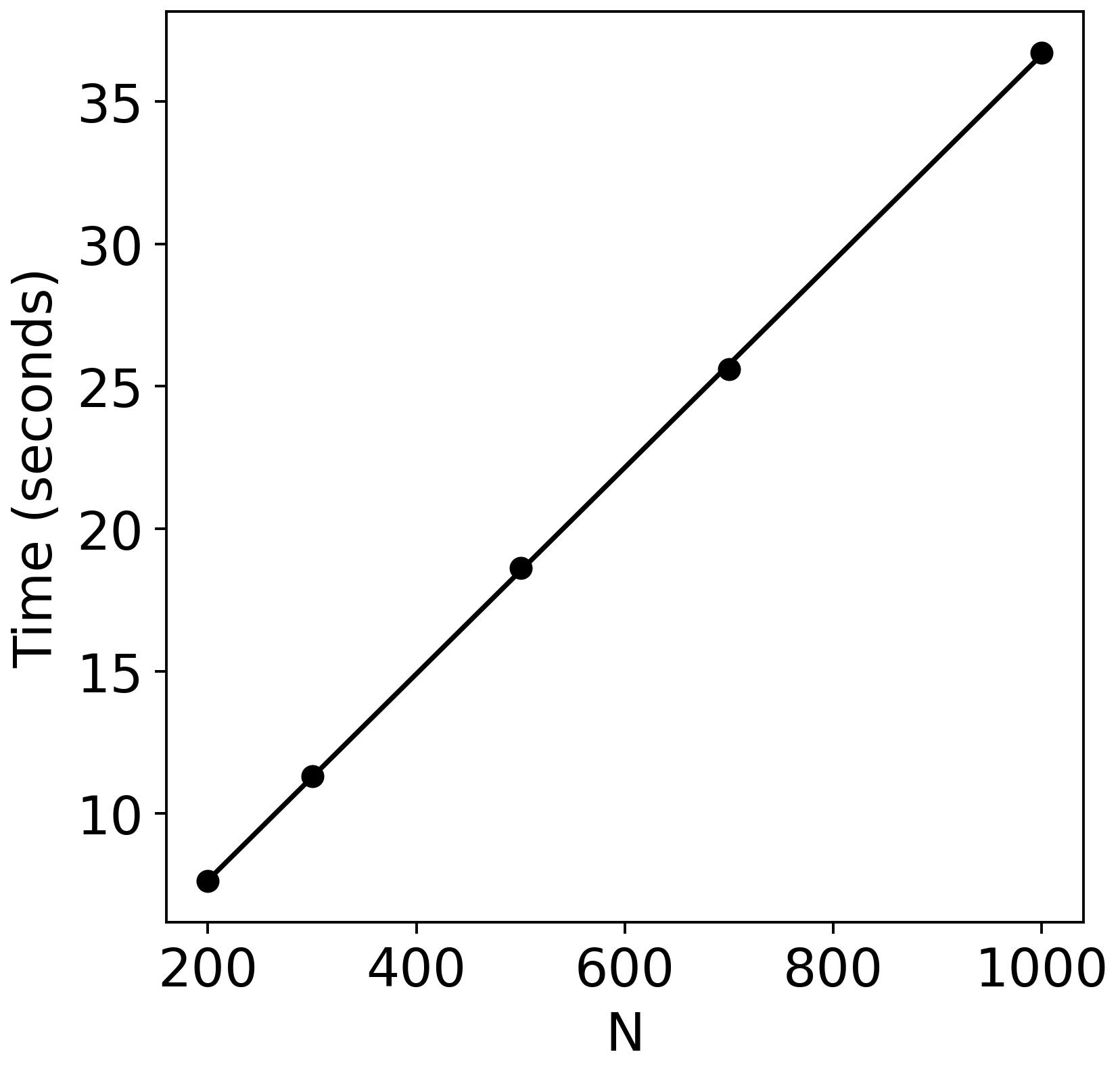}
    \caption{}
\end{subfigure}
\hfill
\begin{subfigure}{.24\linewidth}
    \includegraphics[width=\linewidth]{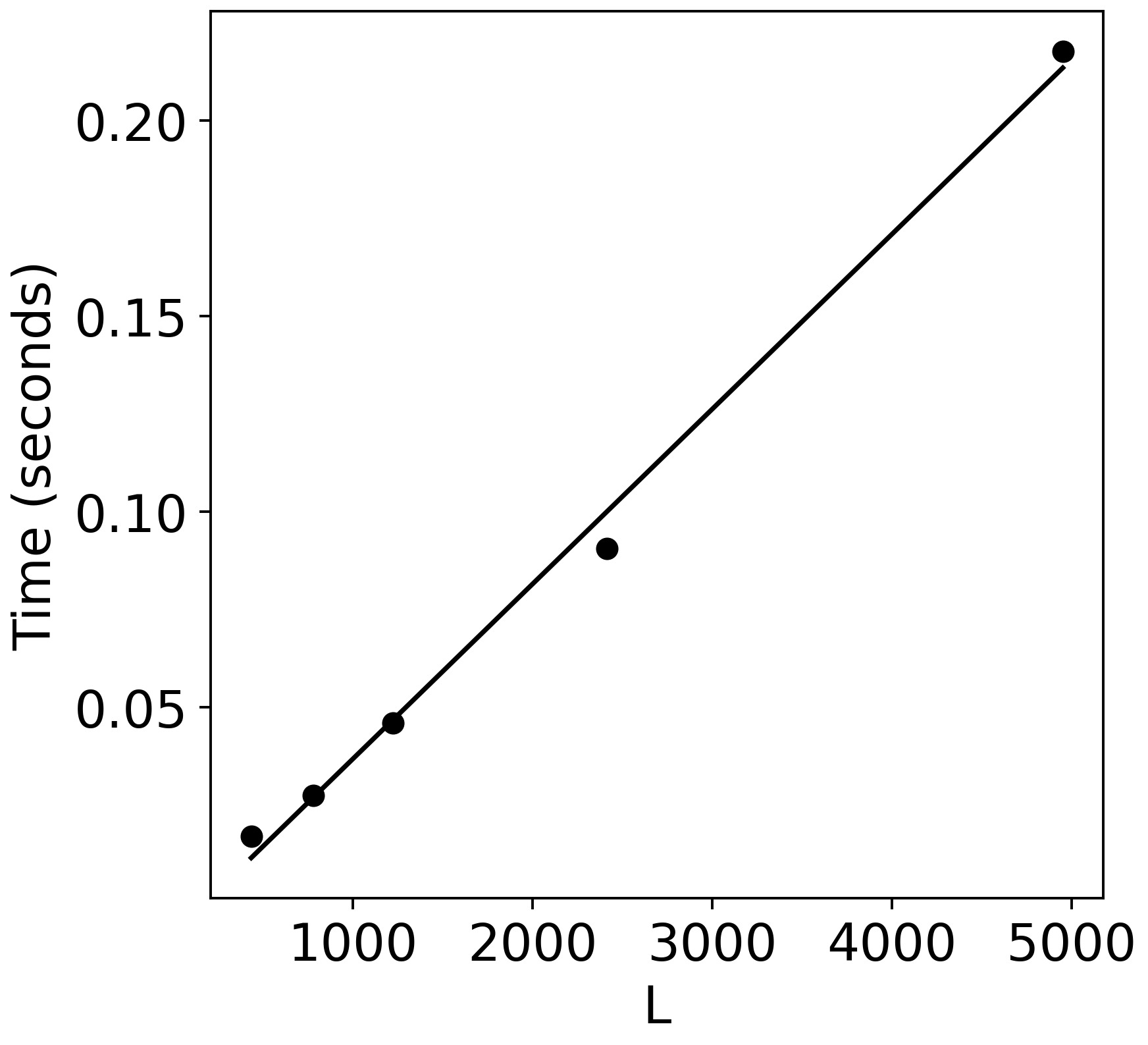}
    \caption{}
\end{subfigure}
\hfill
\begin{subfigure}{.24\linewidth}
    \includegraphics[width=\linewidth]{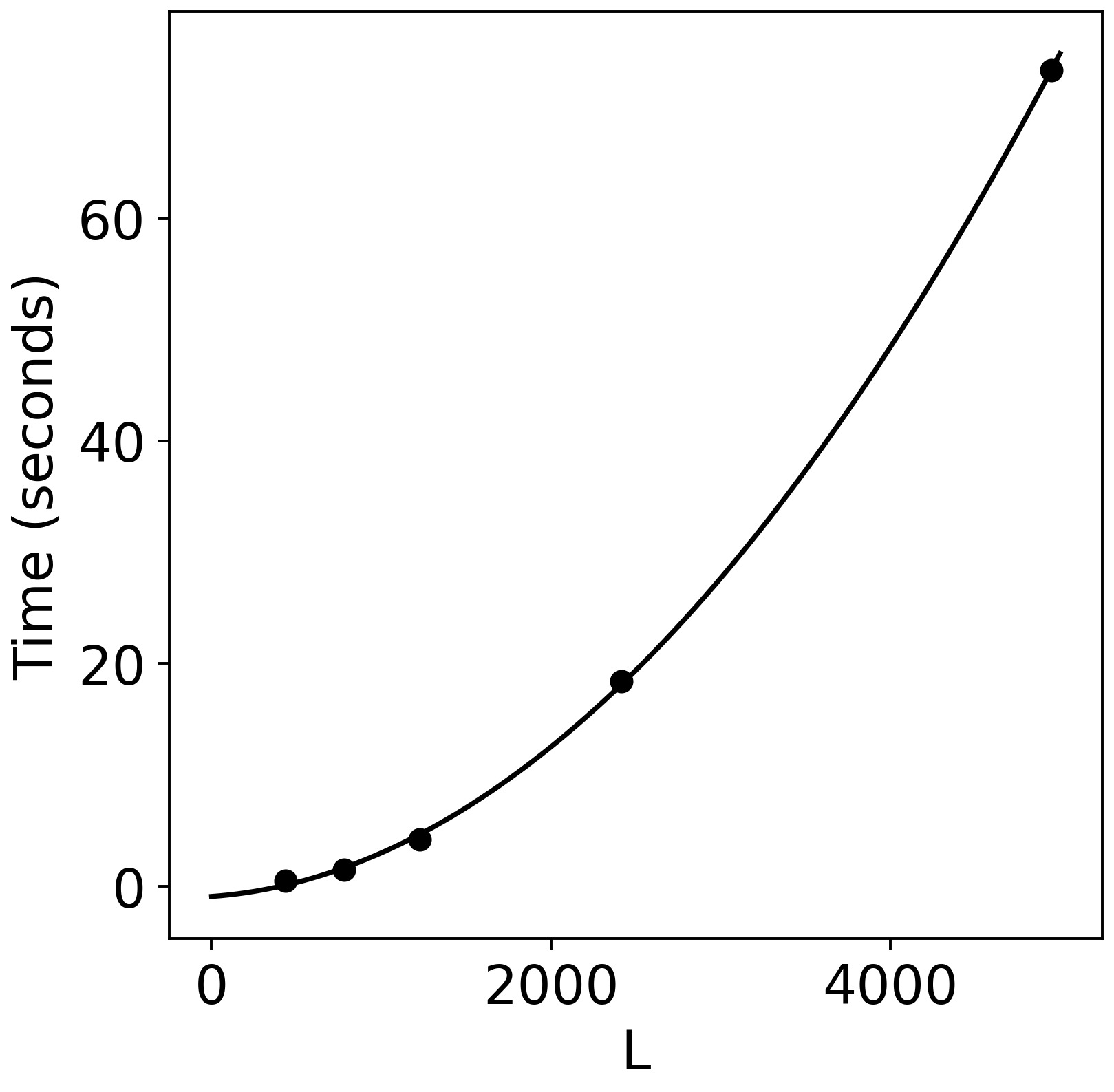}
    \caption{}
\end{subfigure}
\caption{Run-time (in seconds) for (a) each iteration of the estimation algorithm with respect to sample size, $N$; (b) calculation of the influence measures with respect to sample size, $N$; (c) each iteration of the estimation algorithm with respect to number of edges, $L = V(V-1)/2$ and (d) calculation of the influence measures with respect to the number of edges $L$. The first 3 are linear and the last one is quadratic.}
    \label{runtimegraph}
\end{figure}

\subsection{Performance in detecting outliers}

The goal in this section is to assess the performance of ODIN in terms of detection of outliers with varying levels of extremeness. We do this by generating synthetic data sets containing outliers and non-outliers in two different ways. For both these approaches we provide sensitivity (true positive rate) and specificity (true negative rate) values. 

\subsubsection{Simulation from model \texorpdfstring{\eqref{eq:logit}}{(1)}}

We simulate 500 synthetic binary brain networks having 70 ROIs each, with two hemispheres and five lobes within each hemisphere.  As  data from the UKB contain 68 ROIs using the Desikan atlas, our simulated data are similar in dimensionality. We use model \eqref{eq:logit} to simulate data by generating the vectors $Z$ and $\beta_i$'s from standard Cauchy and Gaussian distributions, respectively. We simulate `outliers' in the sample by randomly flipping a fixed percentage of the edges in 10\% of the generated adjacency matrices. We then apply ODIN on the simulated dataset.

Repeating this process many times for various percentages of edges flipped and noting the sensitivity (true positive rate) and specificity (true negative rate) gives us table \ref{senspe_our}. Figure \ref{boxplots} shows box plots of the influence measures for non-outliers and outliers for some choices of percentages of edges flipped. As we can see, even for changes in just 7\% of the edges, we start to notice that the influence measure tends to separate outliers from non-outliers. Table \ref{senspe_our} also shows that sensitivity and specificity values are very close to 100\% indicating our method's efficiency in outlier detection.

\begin{table}
\centering
\begin{tabular}{| >{\centering\arraybackslash}p{0.3\linewidth} | >{\centering\arraybackslash} p{0.3\linewidth}| >{\centering\arraybackslash}p{0.3\linewidth}|}
  \hline
Proportion of flipped edges & Average Sensitivity (5 repetitions) & Average Specificity (5 repetitions) \\ 
  \hline
1\% & 93\% & 94\% \\ 
2\% & 98.2\% & 95.6\% \\ 
7\% & 100\% & 97.1\% \\ 
10\% & 100\% & 97.1\% \\ 
15\% & 100\% & 98\% \\ 
   \hline
\end{tabular}
\vspace*{0.5cm}
\caption{Sensitivity and specificity of the outlier detection algorithm ODIN for network data simulated from model \eqref{eq:logit}.}
\label{senspe_our}
\end{table}

\begin{figure*}
\begin{subfigure}{.33\linewidth}
    \includegraphics[width=\linewidth]{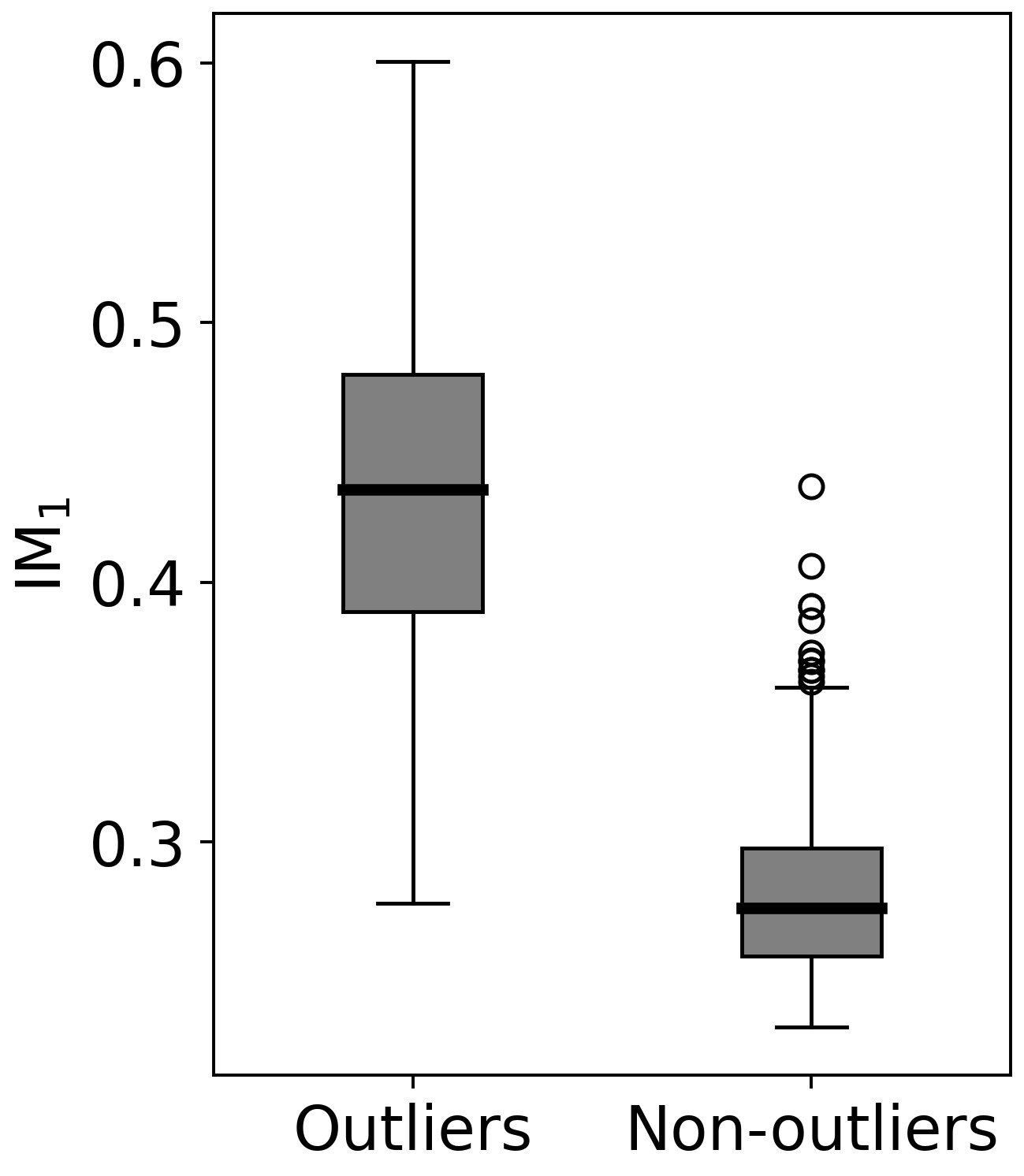}
\end{subfigure}
\hfill
\begin{subfigure}{.33\linewidth}
    \includegraphics[width=\linewidth]{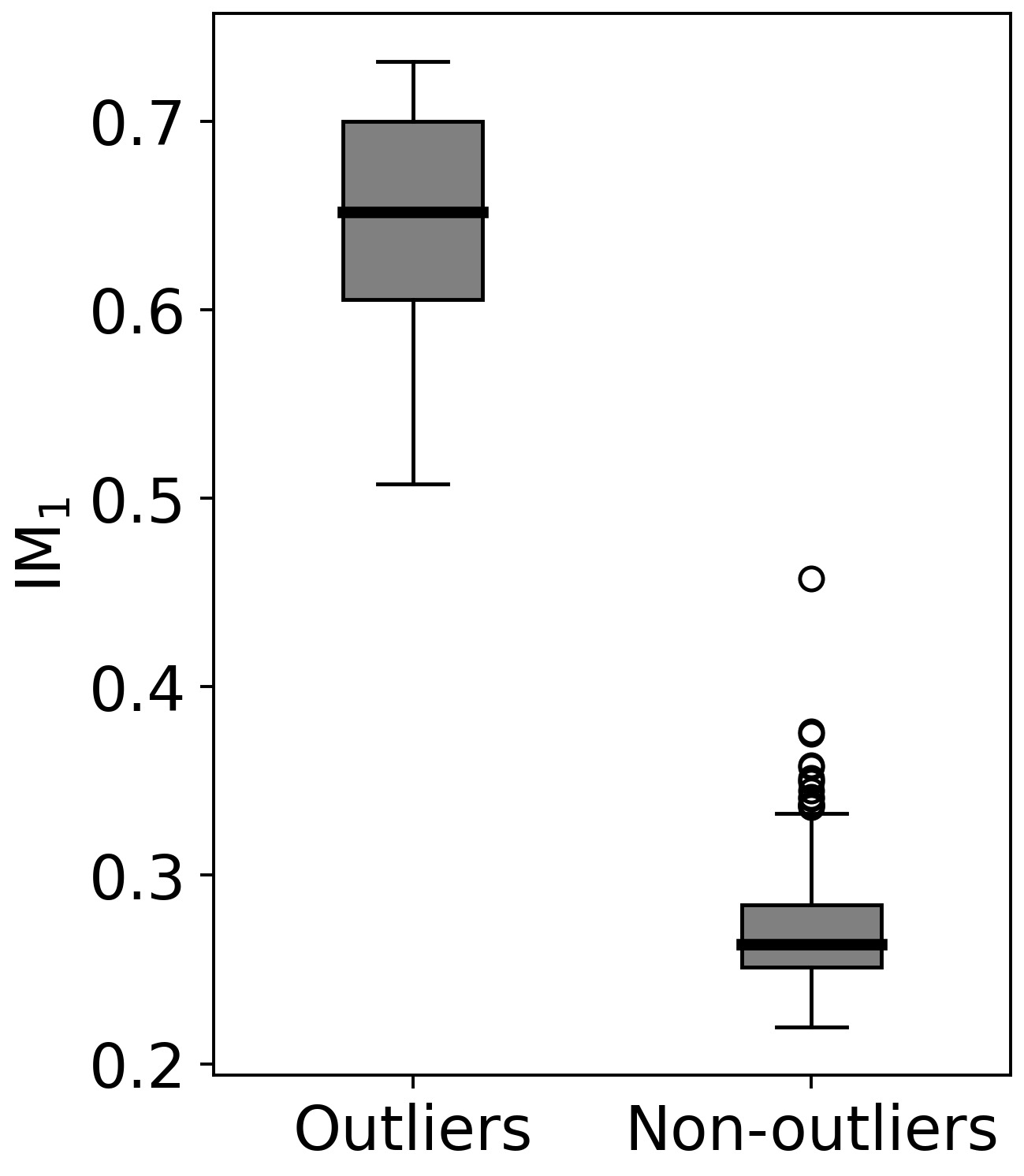}
\end{subfigure}
\hfill
\begin{subfigure}{.33\linewidth}
    \includegraphics[width=\linewidth]{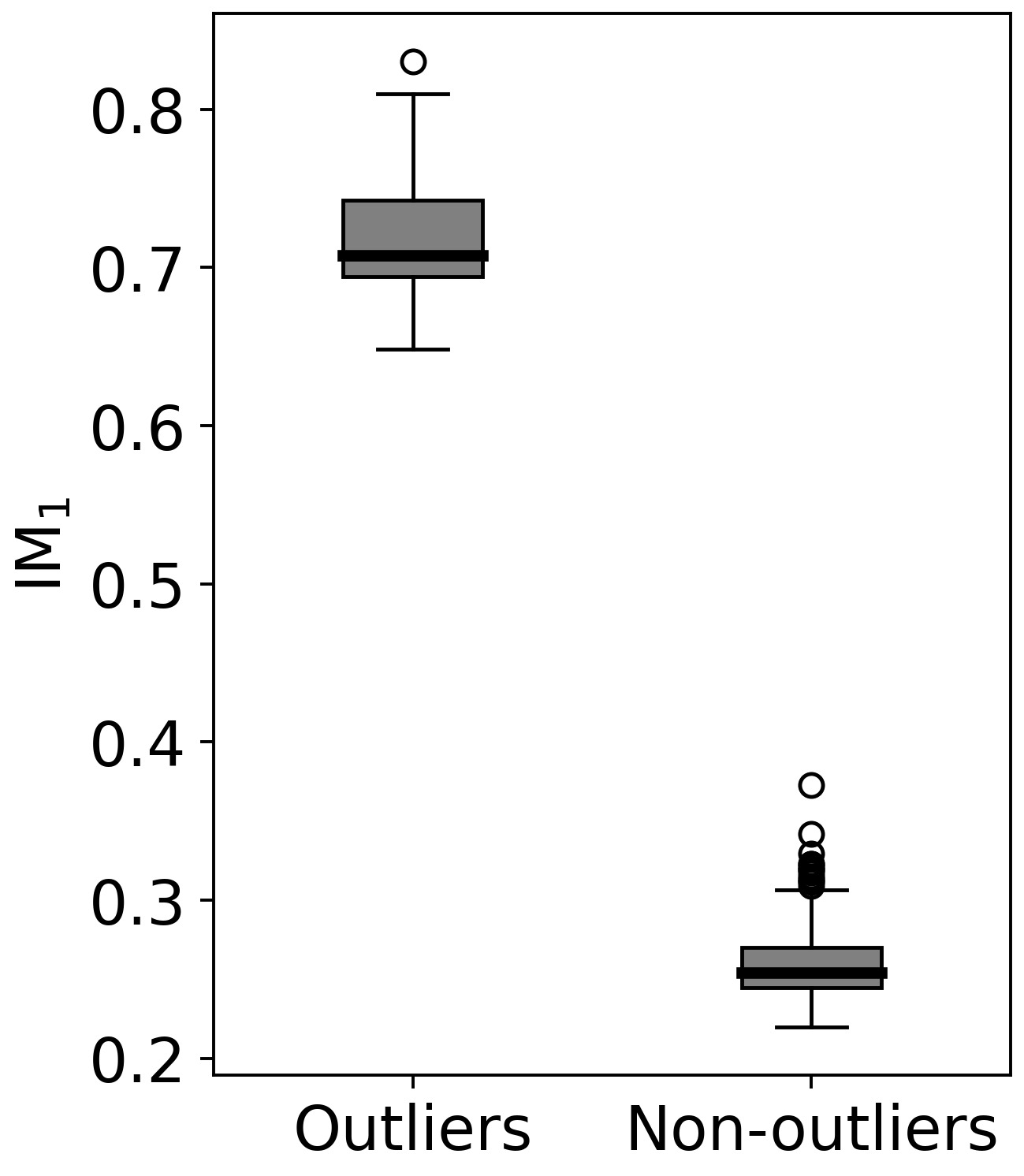}
\end{subfigure}
\caption{Box-plots of $IM_1(i)$ of outliers and non-outliers for data generated from model \eqref{eq:logit} with `outliers' simulated by flipping a fixed proportion of edges for some of the networks. The three figures are for 1\%, 5\% and 10\% of flipped edges respectively. This demonstrates that as outliers become more severe, ODIN can more easily distinguish outliers and non-outliers. }
\label{boxplots}
\end{figure*}

\subsubsection{Simulation using Tensor Network PCA}

To assess how well ODIN can detect outliers in practice, it is important to also measure performance under a very different generative model than \eqref{eq:logit}.
For this, we choose Tensor Network Principal Component Analysis \citep{zhang2019tensor} (TN-PCA) as the data generation scheme. Briefly, TN-PCA is an extension of standard principal component analysis (PCA) to the case of multi-network data. Just as standard PCA reduces dimensionality of vector valued data by embedding the data points in a lower dimensional space, TN-PCA does the same for data in the form of symmetric adjacency matrices. Essentially it minimizes the quantity $\sum_{i=1}^N \norm{A_i - \sum_{k=1}^K \lambda_{ik} \bm{v}_k \bm{v}_k^T }_2^2$ where the $\bm{v}_k$s are unit vectors and are pairwise orthogonal. The vector $\bm{\lambda}_i = (\lambda_{i1},\lambda_{i2},\hdots,\lambda_{iK})$ can be treated as an embedding of the adjacency matrix $A_i$ in $K$-dimensional Euclidean space. In this description, we used slightly different notation than the original paper.

Since TN-PCA is not a generative model explicitly, but rather an embedding method, we use the following steps to simulate data which resemble the UKB brain networks. First we take a small subsample of size 500 from the UKB data, with this 
subsample not containing any outliers (based on ODIN).  Then we perform TN-PCA on the adjacency matrices from that subsample. This gives us vectors $\bm{v}_k$ for $1\leq k \leq K$, and $\bm{\lambda}_i$ for $1 \leq i \leq 500$. We modify these $\bm{\lambda}_i$ as $\widetilde{\bm{\lambda}}_i = \bm{\lambda}_i + \delta_i \bm{\epsilon}_i$ where $\delta_i = 1$ for a randomly selected 10\% of the $i$'s and $\delta_i = 0$ for the rest and $\bm{\epsilon}_i \sim Normal(0,\sigma^2 I)$. We then construct adjacency matrices $\widetilde{A}_i$ by first calculating $S_i = \sum_{k=1}^K \widetilde{\lambda}_{ik} \bm{v}_k \bm{v}_k^T$ and then setting each entry in $\widetilde{A}_i$ to 0 or 1 depending on whichever is closer to the corresponding entry of $S_i$. 

In this way of simulation, the 10\% of the subjects for which $\delta_i=1$ can be thought of as `outliers'. We run ODIN on this sample and repeat this for many choices of $\sigma$, the standard deviation of the Gaussian noise, and note the sensitivity (true positive rate) and specificity (true negative rate) in table \ref{senspe_tnpca}. These results suggest that ODIN is quite efficient in detecting outliers and can easily distinguish them from non-outliers.

\begin{table}
\centering
\begin{tabular}{| >{\centering\arraybackslash}p{0.3\linewidth} | >{\centering\arraybackslash} p{0.3\linewidth}| >{\centering\arraybackslash}p{0.3\linewidth}|}
  \hline
SD of Gaussian Noise & Average Sensitivity (5 repetitions) & Average Specificity (5 repetitions) \\ 
  \hline
0.010 & 71\% & 91.2\% \\ 
0.015 & 97\% & 92.7\% \\ 
0.017 & 98\% & 92.6\% \\ 
0.020 & 100\% & 93\% \\ 
0.025 & 100\% & 92.4\% \\ 
   \hline
\end{tabular}
\vspace*{0.5cm}
\caption{Sensitivity and specificity of the outlier detection algorithm ODIN for network data simulated using TN-PCA.}
\label{senspe_tnpca}
\end{table}

\section{Application to UK Biobank data}

We applied ODIN to a sample of 18,083 brain networks from the UK Biobank (UKB) dataset. We use the Desikan \citep{desikan2006automated} atlas based ROI representation of the brain networks. Although this dataset contains information on hundreds of thousands of subjects, we focus on data for the 18,083 individuals having relevant data to construct adjacency matrices for structural brain connection networks.
Of these, we detected 1,931 (around 10.6\%) networks as outliers.  

\subsection{Exploratory view of outliers vs. non-outliers}

We have already displayed some of the detected outliers in figure \ref{adjfigintro} in the introduction section. This figure shows one possible difference between outliers and non-outliers may be in the number of edges which are connected. Indeed when we plot the distributions of connected edges among outliers and non-outliers in figure \ref{numberedgesconnected} it confirms that many of the outliers have very few or too many connected edges compared to non-outliers.

\begin{figure*}
\begin{subfigure}{0.47\linewidth}
    \includegraphics[width=\linewidth]{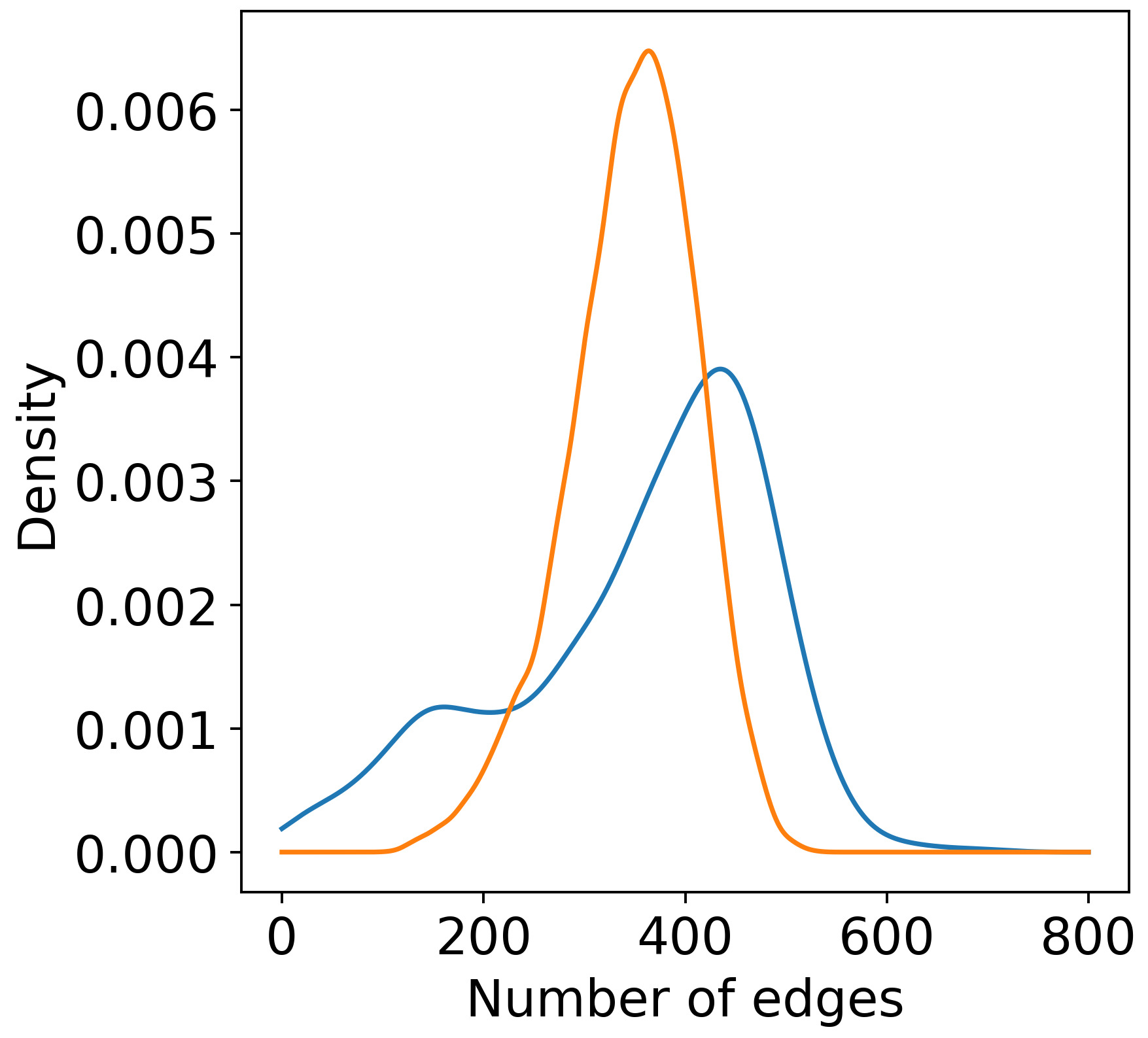}
    \caption{Distribution of number of edges connecting two ROIs located in different hemispheres.}
\end{subfigure}\hfill
\begin{subfigure}{0.47\linewidth}
    \includegraphics[width=\linewidth]{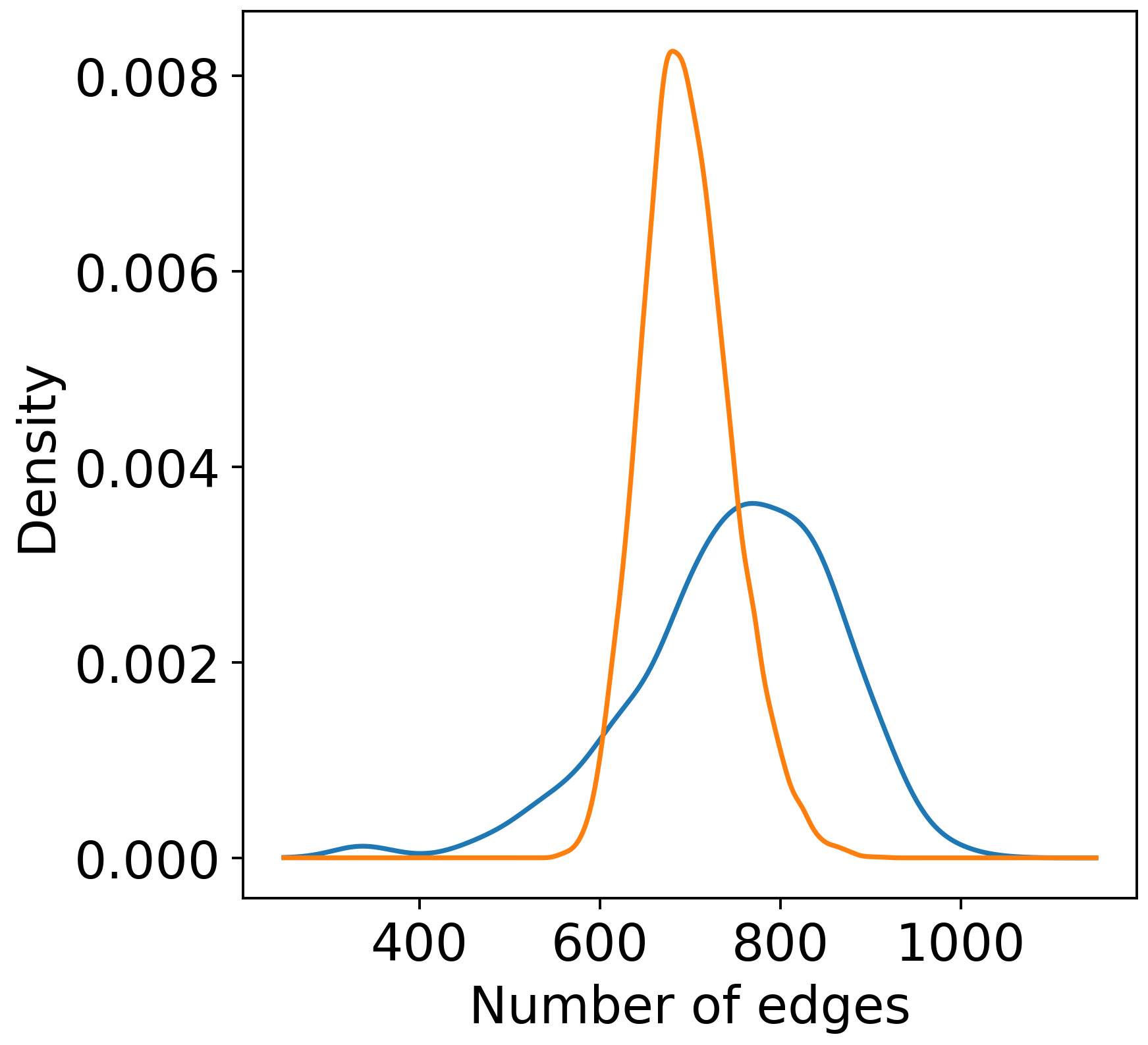}
    \caption{Distribution of number of edges connecting two ROIs located in the same hemisphere.}
\end{subfigure}
\caption{An exploratory view of the difference between outliers and non-outliers detected by ODIN. In each of the figures, the blue graphs represent outliers and the orange graph represents non-outliers.}
\label{numberedgesconnected}
\end{figure*}

\subsection{Impact of sample size on outlier detection}

ODIN detects outliers based on a sample from a population of networks. This raises the question of the extent to which the sample size impacts performance in outlier detection.
To investigate this question, we took stratified subsamples of several different sizes with each subsample containing 10\% outliers (as detected in the full sample) and 90\% non-outliers. We applied ODIN separately to each of these subsamples.
We were interested to see how many of these outliers are also detected as outliers in the subsamples. We were also interested to see how many non-outliers in the full sample were detected as outliers in a subsample. The results from this experiment are presented table \ref{popanalysis}. It appears that even over various different ranges of subsample size, ODIN still classifies most subjects classified as outliers in the full sample as outliers in the subsample while also classifying very few non-outliers from the full sample as outliers in a subsample.

\begin{table*}
\centering
\begin{tabular}{| >{\centering\arraybackslash}p{0.1\textwidth}| >{\centering\arraybackslash}p{0.2\textwidth}| >{\centering\arraybackslash}p{0.2\textwidth}| >{\centering\arraybackslash}p{0.2\textwidth}| >{\centering\arraybackslash}p{0.2\textwidth}|}
\hline
Subsample Size &
Number of outliers from full sample included in subsample &
Number of observations in subsample classified as outliers &
Number of subsample outliers which were outliers in full sample &
Number of non-outliers in full sample classified as outliers in subsample \\ \hline
200  & 20  & 19  & 14 (70\%)    & 5 (1.04\%)  \\
500  & 50  & 52  & 42 (84\%)    & 10 (2.22\%) \\
1000 & 100 & 105 & 88 (88\%)    & 17 (1.89\%) \\
2000 & 200 & 200 & 186 (93\%)   & 14 (0.08\%) \\
5000 & 500 & 509 & 481 (96.2\%) & 28 (0.06\%) \\ \hline
\end{tabular}
\vspace*{0.5cm}
\caption{This table demonstrates the dependence of outlier detection on sample size. ODIN was applied to stratified subsamples of the full sample to see if subjects detected as outliers in the full sample are also detected as outliers in subsamples of different sizes. The percentage values in the fourth column is the percentage of full sample outliers classified as outliers in the subsample. The percentage value in the fifth column are the percentage of full sample non-outliers which are classified as outliers in the subsample.}
\label{popanalysis}
\end{table*}

\subsection{Case study: Impact of removing outliers on a subsequent analysis using TN-PCA}

In addition to the brain network data, the UKB data also contain measurements of traits of various types including substance use, physical characteristics, cognitive abilities, levels of physical activity and measures of mental health. Some of the trait data are missing for some subjects. In this subsection, we assess the impact of removing outliers on inferences relating brain networks to traits.

To illustrate how outliers can impact brain network inferences, including assessments of relationships with cognitive traits, we focus on the impact of outliers on the TN-PCA algorithm. We perform TN-PCA on both the full sample and the full sample with all outliers removed. We then reconstruct the networks from the TN-PCA components and study the differences in connectivity of each edge in the reconstructed networks between high and low scoring individuals for several different trait scores. 

Doing this provides us with a picture of how different the vector representations obtained from TN-PCA of the networks are between subjects with high and low scorers with respect to a specific cognitive trait. We do this for 3 specific traits: numeric memory, symbol digit substitution and fluid intelligence.  

\subsubsection{Numeric memory}

In this test, the participants were shown a 2-digit number to remember. The number then disappeared and after a short while they were asked to enter the number onto the screen. The number became one digit longer each time they remembered correctly (up to 12 digits). The score is the maximum number of digits the participant remembered correctly. This is to assess short term memory.

\begin{figure*}
    \centering
\includegraphics[width=0.7\linewidth]{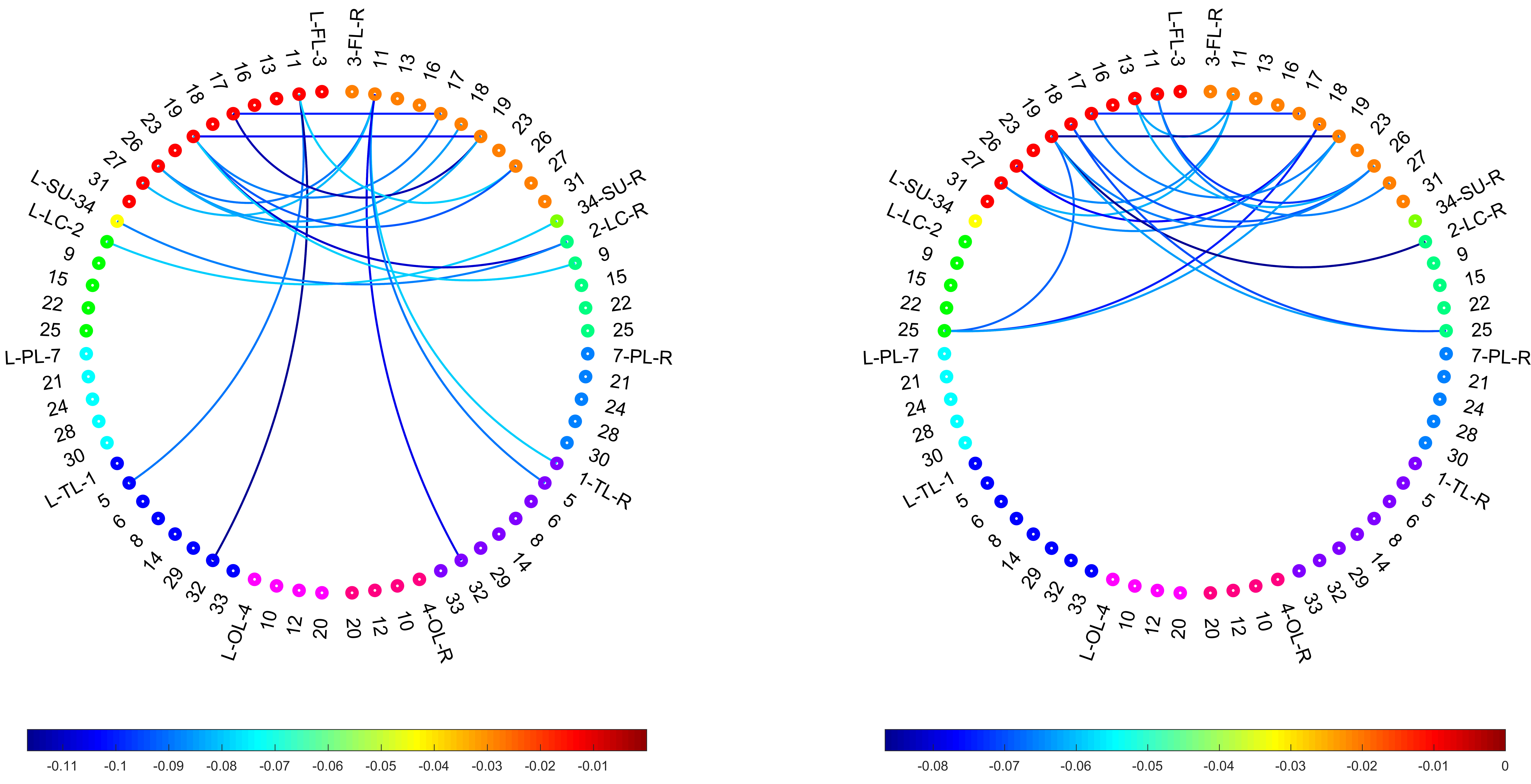}
    \caption{Changes in brain connectivity with increasing numeric memory scores with (left) and without (right) outliers. The 20 edges having the most change in connectivity are shown to improve visualization. The dots on the circle boundary represent ROIs and are colour coded to represent which lobe and hemisphere it belongs to. For a more complete description of the ROI labels see section S3 of the supplement.}
    \label{Nummemchange}
\end{figure*}

Figure \ref{Nummemchange} shows the difference between high and low scorers in this test. In this case, the pattern of differences in connectivity between the groups is essentially completely changed with removal of outliers. After removal of outliers, the resulting pattern strongly suggests that high and low scoring individuals differ primarily in the connectivity in the frontal lobe. This is in agreement with previous knowledge \citep{jacobsen1936studies,funahashi1993dorsolateral,pribram1952effects} that the frontal lobe, in particular the prefrontal cortex, plays a major role in short term memory.

\subsubsection{Symbol digit substitution}

In this test, participants were presented with a series of grids in which symbols were to be matched to numbers according to a key presented on the screen. We consider the number of symbols correctly matched by each participant. Again the differences in brain connectivity for the 10\% lowest and 10\% highest scorers is shown in figure \ref{Symbolchange}. In this case removing outliers seems to significantly change the detected difference in connectivity among high and low scorers in this trait.

\begin{figure*}
    \centering
\includegraphics[width=0.7\linewidth]{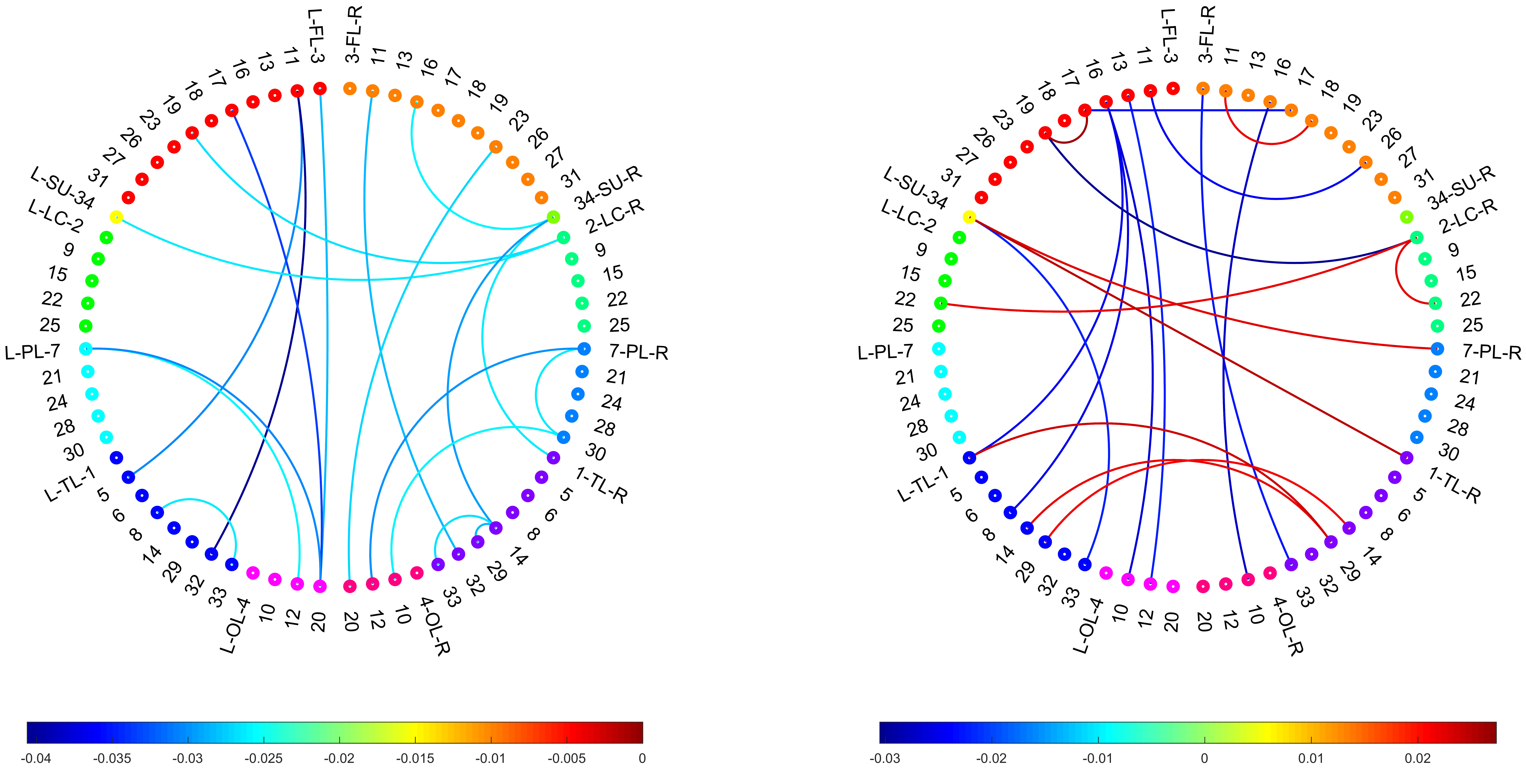}
    \caption{Changes in brain connectivity with increasing symbol digit substitution scores with (left) and without (right) outliers. The 20 edges having the most change in connectivity are shown to improve visualization. The dots on the circle boundary represent ROIs and are colour coded to represent which lobe and hemisphere it belongs to. For a more complete description of the ROI labels see section S3 of the supplement.}
    \label{Symbolchange}
\end{figure*}

\subsubsection{Fluid intelligence}

Participants were asked 13 questions designed to assess `Fluid intelligence', i.e., the capacity to solve problems that require logic and reasoning ability, independent of pre-acquired knowledge. Each participant was given 2 minutes to complete as many questions as possible from the test. Figure \ref{Fluidchange} shows the differences in brain connectivity with and without outliers among the 10\% lowest and 10\% highest scorers for the fluid intelligence test. In this case, removing outliers has some impact on the inferred differences, but the impact is more subtle than in the previous two examples. In both cases fluid intelligence seems to be highly related to connections in the frontal lobe, a fact that seems to line up with literature \citep{geake2005neural} on the subject.

\begin{figure*}
    \centering
    \includegraphics[width=0.7\linewidth]{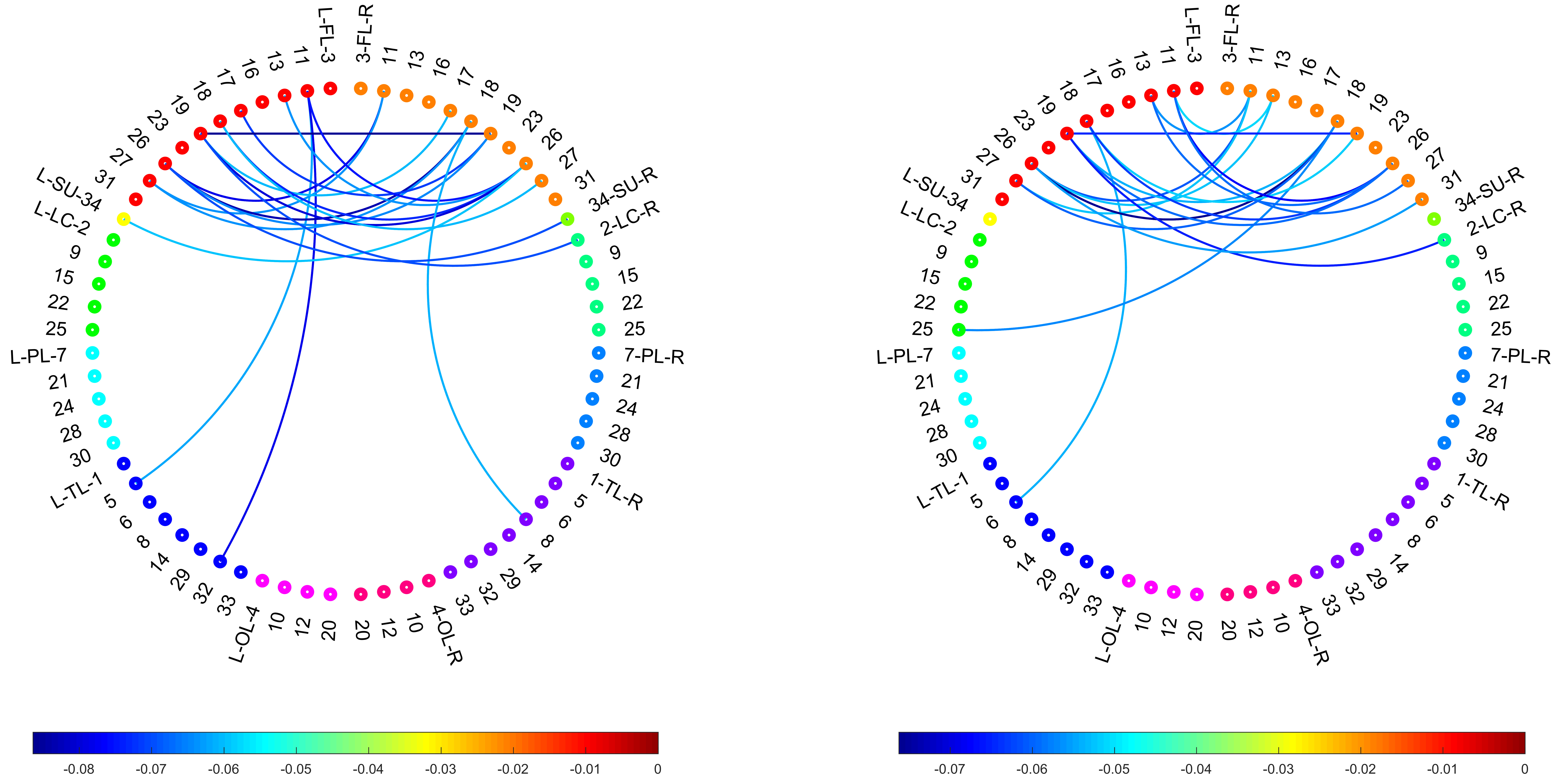}
    \caption{Changes in brain connectivity with increasing fluid intelligence scores with (left) and without (right) outliers. The 20 edges having the most change in connectivity are shown to improve visualization. The dots on the circle boundary represent ROIs and are colour coded to represent which lobe and hemisphere it belongs to. For a more complete description of the ROI labels see section S3 of the supplement.}
    \label{Fluidchange}
\end{figure*}

The above three case studies demonstrate that outliers can greatly influence inferences on how brain networks relate to traits, to the extent of fundamentally altering
scientific conclusions.

\section{Discussion}

In this article, we propose a fast, simple and effective approach to the problem of outlier detection for brain structural connectivity data represented as binary (i.e., unweighted) adjacency matrices. Our method, ODIN, involves first fitting a simple model based on logistic regression which contains both population and subject level parameters. Then using the fitted model, ODIN measures outlyingness based on: (1) approximating the change in the pMLE of the population parameters when a subject is dropped and (2) tracking how extreme the subject-specific parameters are. The resulting measures of influence are then thresholded to classify subjects as outliers/non-outliers. Using different simulations and an application on a large dataset, we demonstrated that ODIN is fast and effective in detecting outliers. We also illustrated how removing outliers can fundamentally impact scientific conclusions from brain network studies.

There are several natural next steps building on the ODIN approach.  An important extension is to generalize ODIN to accommodate weighted adjacency matrices.  If the weights are modeled parametrically via a generalized linear model, then this extension is very straightforward.  For example, if the weights consist of the number of fibers connecting each pair of ROIs, then we could use a negative binomial log-linear model in place of the logistic regression.  If the weights are instead continuous, we could instead use a Gaussian linear model. This modification makes ODIN directly applicable to outlier detection for functional connectivity (FC).  For example, FC networks are often expressed in a weighted form as correlation matrices. 
Applying a Fisher transformation to these correlations, one can use the same linear predictor as in \eqref{eq:logit} but within a Gaussian linear model instead of a logistic regression.  

A more ambitious next step is to detect outliers in SC and FC simultaneously, potentially leveraging on a multivariate extension of model \eqref{eq:logit} which allows for SC-FC dependence. In addition, the distribution of the weights may not be well characterized by simple parametric GLMs and it becomes interesting to develop methods that accommodate flexible distributions for the weights.

\section*{Acknowledgements}

This research has been conducted using the UK Biobank Resource under application number 51659.

\section*{Funding}
This work was supported by the National Institutes of Mental Health [MH118927, AG066970] of the United States National Institutes of Health.

\section*{Data Availability}
The data underlying this article were provided by UK Biobank under licence / by permission. Data will be shared on request to the corresponding author with permission of UK Biobank.

\bibliographystyle{unsrtnat}
\bibliography{references}

\begin{thebibliography}{22}
\providecommand{\natexlab}[1]{#1}
\providecommand{\url}[1]{\texttt{#1}}
\expandafter\ifx\csname urlstyle\endcsname\relax
  \providecommand{\doi}[1]{doi: #1}\else
  \providecommand{\doi}{doi: \begingroup \urlstyle{rm}\Url}\fi

\bibitem[Van~Essen et~al.(2013)Van~Essen, Smith, Barch, Behrens, Yacoub,
  Ugurbil, Consortium, et~al.]{van2013wu}
David~C Van~Essen, Stephen~M Smith, Deanna~M Barch, Timothy~EJ Behrens, Essa
  Yacoub, Kamil Ugurbil, {WU-Minn}~{HCP} Consortium, et~al.
\newblock The {WU-Minn} human connectome project: an overview.
\newblock \emph{NeuroImage}, 80:\penalty0 62--79, 2013.

\bibitem[Casey et~al.(2018)Casey, Cannonier, Conley, Cohen, Barch, Heitzeg,
  Soules, Teslovich, Dellarco, Garavan, et~al.]{casey2018adolescent}
BJ~Casey, Tariq Cannonier, May~I Conley, Alexandra~O Cohen, Deanna~M Barch,
  Mary~M Heitzeg, Mary~E Soules, Theresa Teslovich, Danielle~V Dellarco, Hugh
  Garavan, et~al.
\newblock The adolescent brain cognitive development {(ABCD)} study: imaging
  acquisition across 21 sites.
\newblock \emph{Developmental Cognitive Neuroscience}, 32:\penalty0 43--54,
  2018.

\bibitem[Miller et~al.(2016)Miller, Alfaro-Almagro, Bangerter, Thomas, Yacoub,
  Xu, Bartsch, Jbabdi, Sotiropoulos, Andersson, et~al.]{miller2016multimodal}
Karla~L Miller, Fidel Alfaro-Almagro, Neal~K Bangerter, David~L Thomas, Essa
  Yacoub, Junqian Xu, Andreas~J Bartsch, Saad Jbabdi, Stamatios~N Sotiropoulos,
  Jesper~LR Andersson, et~al.
\newblock Multimodal population brain imaging in the {UK Biobank} prospective
  epidemiological study.
\newblock \emph{Nature Neuroscience}, 19\penalty0 (11):\penalty0 1523--1536,
  2016.

\bibitem[Wang et~al.(2019)Wang, Zhang, Dunson, et~al.]{wang2019common}
Lu~Wang, Zhengwu Zhang, David Dunson, et~al.
\newblock Common and individual structure of brain networks.
\newblock \emph{The Annals of Applied Statistics}, 13\penalty0 (1):\penalty0
  85--112, 2019.

\bibitem[Durante et~al.(2017)Durante, Dunson, and
  Vogelstein]{durante2017nonparametric}
Daniele Durante, David~B Dunson, and Joshua~T Vogelstein.
\newblock Nonparametric {Bayes} modeling of populations of networks.
\newblock \emph{Journal of the American Statistical Association}, 112\penalty0
  (520):\penalty0 1516--1530, 2017.

\bibitem[Aliverti and Durante(2019)]{aliverti2019spatial}
Emanuele Aliverti and Daniele Durante.
\newblock Spatial modeling of brain connectivity data via latent distance
  models with nodes clustering.
\newblock \emph{Statistical Analysis and Data Mining: The ASA Data Science
  Journal}, 12\penalty0 (3):\penalty0 185--196, 2019.

\bibitem[Ginestet et~al.(2017)Ginestet, Li, Balachandran, Rosenberg, Kolaczyk,
  et~al.]{ginestet2017hypothesis}
Cedric~E Ginestet, Jun Li, Prakash Balachandran, Steven Rosenberg, Eric~D
  Kolaczyk, et~al.
\newblock Hypothesis testing for network data in functional neuroimaging.
\newblock \emph{The Annals of Applied Statistics}, 11\penalty0 (2):\penalty0
  725--750, 2017.

\bibitem[Wang et~al.(2017)Wang, Durante, Jung, and Dunson]{wang2017bayesian}
Lu~Wang, Daniele Durante, Rex~E Jung, and David~B Dunson.
\newblock Bayesian network--response regression.
\newblock \emph{Bioinformatics}, 33\penalty0 (12):\penalty0 1859--1866, 2017.

\bibitem[Zhang et~al.(2022)Zhang, Sun, and Li]{zhang2018network}
Jingfei Zhang, Will~Wei Sun, and Lexin Li.
\newblock {Generalized Connectivity Matrix Response Regression with
  Applications in Brain Connectivity Studies}.
\newblock \emph{Journal of Computational and Graphical Statistics}, 0\penalty0
  (ja):\penalty0 1--30, 2022.

\bibitem[Zhang et~al.(2019)Zhang, Allen, Zhu, and Dunson]{zhang2019tensor}
Zhengwu Zhang, Genevera~I Allen, Hongtu Zhu, and David Dunson.
\newblock Tensor network factorizations: Relationships between brain structural
  connectomes and traits.
\newblock \emph{NeuroImage}, 197:\penalty0 330--343, 2019.

\bibitem[Fornito et~al.(2013)Fornito, Zalesky, and
  Breakspear]{fornito2013graph}
Alex Fornito, Andrew Zalesky, and Michael Breakspear.
\newblock Graph analysis of the human connectome: promise, progress, and
  pitfalls.
\newblock \emph{NeuroImage}, 80:\penalty0 426--444, 2013.

\bibitem[Baum et~al.(2018)Baum, Roalf, Cook, Ciric, Rosen, Xia, Elliott,
  Ruparel, Verma, Tun{\c{c}}, et~al.]{baum2018impact}
Graham~L Baum, David~R Roalf, Philip~A Cook, Rastko Ciric, Adon~FG Rosen,
  Cedric Xia, Mark~A Elliott, Kosha Ruparel, Ragini Verma, Birkan Tun{\c{c}},
  et~al.
\newblock The impact of in-scanner head motion on structural connectivity
  derived from diffusion {MRI}.
\newblock \emph{NeuroImage}, 173:\penalty0 275--286, 2018.

\bibitem[Zhang et~al.(2018)Zhang, Descoteaux, Zhang, Girard, Chamberland,
  Dunson, Srivastava, and Zhu]{zhang2018mapping}
Zhengwu Zhang, Maxime Descoteaux, Jingwen Zhang, Gabriel Girard, Maxime
  Chamberland, David Dunson, Anuj Srivastava, and Hongtu Zhu.
\newblock Mapping population-based structural connectomes.
\newblock \emph{NeuroImage}, 172:\penalty0 130--145, 2018.

\bibitem[Alfaro-Almagro et~al.(2018)Alfaro-Almagro, Jenkinson, Bangerter,
  Andersson, Griffanti, Douaud, Sotiropoulos, Jbabdi, Hernandez-Fernandez,
  Vallee, et~al.]{alfaro2018image}
Fidel Alfaro-Almagro, Mark Jenkinson, Neal~K Bangerter, Jesper~LR Andersson,
  Ludovica Griffanti, Gwena{\"e}lle Douaud, Stamatios~N Sotiropoulos, Saad
  Jbabdi, Moises Hernandez-Fernandez, Emmanuel Vallee, et~al.
\newblock {Image processing and Quality Control for the first 10,000 brain
  imaging datasets from UK Biobank}.
\newblock \emph{NeuroImage}, 166:\penalty0 400--424, 2018.

\bibitem[Hawkins(1980)]{hawkins1980identification}
Douglas~M Hawkins.
\newblock \emph{Identification of outliers}, volume~11.
\newblock Springer Netherlands, 1980.

\bibitem[Wang and Wedeen(2007)]{wang2007trackvis}
Ruopeng Wang and Van~J Wedeen.
\newblock {TrackVis.org}.
\newblock \emph{Martinos Center for Biomedical Imaging, Massachusetts General
  Hospital}, 2007.

\bibitem[Desikan et~al.(2006)Desikan, S{\'e}gonne, Fischl, Quinn, Dickerson,
  Blacker, Buckner, Dale, Maguire, Hyman, et~al.]{desikan2006automated}
Rahul~S Desikan, Florent S{\'e}gonne, Bruce Fischl, Brian~T Quinn, Bradford~C
  Dickerson, Deborah Blacker, Randy~L Buckner, Anders~M Dale, R~Paul Maguire,
  Bradley~T Hyman, et~al.
\newblock An automated labeling system for subdividing the human cerebral
  cortex on {MRI} scans into gyral based regions of interest.
\newblock \emph{NeuroImage}, 31\penalty0 (3):\penalty0 968--980, 2006.

\bibitem[Satopaa et~al.(2011)Satopaa, Albrecht, Irwin, and Raghavan]{5961514}
Ville Satopaa, Jeannie Albrecht, David Irwin, and Barath Raghavan.
\newblock {Finding a "Kneedle" in a Haystack: Detecting Knee Points in System
  Behavior}.
\newblock In \emph{2011 31st International Conference on Distributed Computing
  Systems Workshops}, pages 166--171, 2011.

\bibitem[Jacobsen(1936)]{jacobsen1936studies}
CF~Jacobsen.
\newblock Studies of cerebral function in primates. i. the functions of the
  frontal association areas in monkeys. comp.
\newblock \emph{Psychol. Monogr}, 13, 1936.

\bibitem[Funahashi et~al.(1993)Funahashi, Bruce, and
  Goldman-Rakic]{funahashi1993dorsolateral}
Shintaro Funahashi, Charles~J Bruce, and Patricia~S Goldman-Rakic.
\newblock Dorsolateral prefrontal lesions and oculomotor delayed-response
  performance: evidence for mnemonic" scotomas".
\newblock \emph{Journal of Neuroscience}, 13\penalty0 (4):\penalty0 1479--1497,
  1993.

\bibitem[Pribram et~al.(1952)Pribram, Mishkin, Rosvold, and
  Kaplan]{pribram1952effects}
Karl~H Pribram, Mortimer Mishkin, H~Enger Rosvold, and Sylvan~J Kaplan.
\newblock Effects on delayed-response performance of lesions of dorsolateral
  and ventromedial frontal cortex of baboons.
\newblock \emph{Journal of comparative and physiological psychology},
  45\penalty0 (6):\penalty0 565, 1952.

\bibitem[Geake and Hansen(2005)]{geake2005neural}
John~G Geake and Peter~C Hansen.
\newblock Neural correlates of intelligence as revealed by {fMRI} of fluid
  analogies.
\newblock \emph{NeuroImage}, 26\penalty0 (2):\penalty0 555--564, 2005.

\end{thebibliography}


\begin{thebibliography}{22}
\providecommand{\natexlab}[1]{#1}
\providecommand{\url}[1]{\texttt{#1}}
\expandafter\ifx\csname urlstyle\endcsname\relax
  \providecommand{\doi}[1]{doi: #1}\else
  \providecommand{\doi}{doi: \begingroup \urlstyle{rm}\Url}\fi

\bibitem[Satopaa et~al.(2011)Satopaa, Albrecht, Irwin, and Raghavan]{5961514}
Ville Satopaa, Jeannie Albrecht, David Irwin, and Barath Raghavan.
\newblock {Finding a "Kneedle" in a Haystack: Detecting Knee Points in System
  Behavior}.
\newblock In \emph{2011 31st International Conference on Distributed Computing
  Systems Workshops}, pages 166--171, 2011.

\end{thebibliography}
\end{document}


\maketitle
\graphicspath{{figures/}}

\newcommand{\beginsupplement}{%
        \setcounter{table}{0}
        \renewcommand{\thetable}{S\arabic{table}}%
        \setcounter{figure}{0}
        \renewcommand{\thefigure}{S\arabic{figure}}%
        \setcounter{section}{0}
        \setcounter{subsection}{0}
        \setcounter{equation}{0}
        \renewcommand{\theequation}{S\arabic{equation}}
        \renewcommand{\thesection}{S\arabic{section}}
}
\beginsupplement

This document contain additional material including mathematical results, algorithms and other details to support the paper `Outlier Detection for Multi-Network Data.' All mathematical notations introduced in the main paper are also valid for this document.

\section{Mathematical Proofs and Results}

\begin{lemma}
(i) \textbf{Inverse of block matrix:} If an $n \times n$ symmetric matrix $M$ is partitioned into four blocks as $M=\begin{bmatrix}
A & B\\
B^T & C 
\end{bmatrix}$, where $A$ and $C$ are invertible square matrices of arbitrary size and the Schur complement of $C$ in $M$ ($A-BC^{-1}B^T$) is invertible, then the inverse of $M$ can be expressed as:
\begin{equation}
M^{-1}=\begin{bmatrix}
A & B\\
B^T & C 
\end{bmatrix}^{-1} = \begin{bmatrix}
(A-BC^{-1}B^T)^{-1} & -(A-BC^{-1}B^T)^{-1}BC^{-1}\\
-C^{-1}B^T(A-BC^{-1}B^T)^{-1} & C^{-1} + C^{-1}B^T(A-BC^{-1}B^T)^{-1}BC^{-1} 
\end{bmatrix}
\end{equation}
(ii) If $\begin{pmatrix} x_1 \\ x_2 \end{pmatrix}= M^{-1} \begin{pmatrix} v_1 \\ v_2 \end{pmatrix}$ then the following formulas hold:
\begin{align}
    x_1 & = (A-BC^{-1}B^T)^{-1}(v_1-BC^{-1}v_2) \label{S1}\\
    x_2 & = C^{-1}(v_2 - B^Tx_1) \label{S2}
\end{align}
\label{invlemma}
\end{lemma}

\begin{proof}
Part (i) is a standard result which can be verified by matrix multiplication of $M$ and $M^{-1}$ in the given forms. Part (ii) follows easily from part (i) by explicit multiplication and simple algebraic manipulations.
\end{proof}

\subsection{MM algorithm for estimation of \texorpdfstring{$(Z,\beta_1,\beta_2,...,\beta_N)$}{Z,b1,b2,...bN}}
Concatenating $(Z,\beta_1,\beta_2,...,\beta_N)$ into a single parameter vector $\theta$ in that order, allows us to calculate the gradient vector and Hessian matrix: 
\begin{equation}
    \nabla \mathcal{L}(Z,\{\beta_j\}_{j=1}^N) = \frac{1}{N} 
    \begin{pmatrix}
    \sum_{i=1}^N (\bm{a}_i-\bm{\pi}_i)\\
    X^T (\bm{a}_1-\bm{\pi}_1)\\
    X^T (\bm{a}_2-\bm{\pi}_2)\\
    \vdots\\
    X^T (\bm{a}_N-\bm{\pi}_N)
    \end{pmatrix}
    -
    \begin{pmatrix}
    \lambda Z \\
    0\\
    0\\
    \vdots\\
    0
    \end{pmatrix}
    \label{gradient_supp}
\end{equation}

\begin{equation}
    \mathcal{H}(Z,\{\beta_j\}_{j=1}^N) = -\frac{1}{N} \begin{bmatrix}
\sum_{i=1}^N W_i &\vline & W_1^TX  & \hdots &W_N^TX \\
\hline
X^TW_1 &\vline & X^TW_1X &  \hdots & 0\\

\vdots &\vline & \vdots &\ddots &\vdots\\
X^TW_N &\vline & 0  & \hdots & X^TW_NX \\
\end{bmatrix}-\lambda \begin{bmatrix}
I &\vline & 0  & \hdots &0 \\
\hline
0 &\vline & 0 & \hdots & 0\\
\vdots &\vline & \vdots & \ddots & \vdots\\
0 &\vline & 0 & \hdots & 0 \\
\end{bmatrix}
\label{hessian_supp}
\end{equation}

The first matrix in the equation (\ref{hessian_supp}) can be expressed as $\Tilde{X}^TW\Tilde{X}$ where $W$ is a diagonal matrix consisting of the matrices $W_1,W_2,...,W_N$ stacked diagonally into a single matrix and
\begin{equation}
    \Tilde{X}=\begin{bmatrix}
I &\vline & X &0 & \hdots &0 \\
I &\vline & 0 &X & \hdots & 0\\
\vdots &\vline &\vdots & \vdots & \ddots & \vdots\\
I &\vline & 0 & 0 & \hdots & X \\
\end{bmatrix}
\end{equation}

Further we note that every entry in the diagonal matrix $W$ is bounded above by $\frac{1}{4}$. So, the difference $\frac{1}{4}\Tilde{X}^T\Tilde{X}-\Tilde{X}^TW\Tilde{X}$ is non-negative definite which gives us the following minorizing function for MM algorithm,
\begin{equation}
    g\left(\theta\mid\theta^{(k)}\right) = \mathcal{L}\left(\theta^{(k)}\right) + \nabla \mathcal{L}\left(\theta^{(k)}\right)^T \left(\theta-\theta^{(k)}\right)+
\frac{1}{2}\left(\theta-\theta^{(k)}\right)^TM\left(\theta-\theta^{(k)}\right)
\end{equation}
where $M$ is obtained by replacing $\frac{1}{4}\Tilde{X}^T\Tilde{X}$ in place of $\Tilde{X}^TW\Tilde{X}$ in the Hessian matrix. So,
\begin{equation}
    M=-\frac{1}{4N}\begin{bmatrix}
(4\lambda+1)NI &\vline & X & X  & \hdots &X \\
\hline
X^T &\vline &X^TX & 0 & \hdots & 0\\
X^T &\vline &0 & X^TX & \hdots & 0\\
\vdots &\vline & \vdots & \vdots & \ddots & \vdots\\
X^T &\vline & 0 & 0 & \hdots & X^TX \\
\end{bmatrix}
\end{equation}

The maximizer of $g\left(\theta\mid\theta^{(k)}\right)$ is then $\theta^{(k+1)}=\theta^{(k)}-M^{-1}\nabla \mathcal{L}\left(\theta^{(k)}\right)$ which means that using part (ii) of the lemma, the updates of $Z$ and $\beta_i$ are:
\begin{align}
       Z^{(k+1)}&=Z^{(k)}+4\left[4\lambda I+I-X(X^TX)^{-1}X^T\right]^{-1}  \nonumber
       \\& \qquad \qquad \qquad \left[-\lambda Z^{(k)}+\left(I-X(X^TX)^{-1}X^T\right)\frac{\sum_{i=1}^N\left(\bm{a}_i-\bm{\pi}_i^{(k)}\right)}{N}\right] \\
       \beta_i^{(k+1)}&=\beta_i^{(k)} + 4(X^TX)^{-1}X^T\left[\left(\bm{a}_i-\bm{\pi}^{(k)}_i\right)-N\left(Z^{(k+1)}-Z^{(k)}\right)\right]
\end{align}

As we are not updating the whole vector at once, we save a lot of computation with large matrices. The $\beta_i$'s are not very large in dimension in general, so the only computation with large matrices that has to be done in every iteration are the matrix multiplications in the update of $Z$.

\SetKwInput{KwInput}{Input}                
\SetKwInput{KwOutput}{Output}              

\begin{algorithm}[H]
\DontPrintSemicolon
  \KwData{Binary adjacency matrices: $A_1,A_2,\hdots,A_N$}
  \KwInput{Hemisphere and lobe memberships for every ROI in the atlas, $\lambda$ and a tolerance value for the stopping condition.}
  \KwOutput{Estimated parameter vectors $(\widehat{Z},\widehat{\beta}_1,\widehat{\beta}_2,\hdots,\widehat{\beta}_N)$}
  
  Vectorize the adjacency matrices by taking only the elements below the diagonal to get vectors $\{\bm{a}_i\}_{i=1}^N$. \;
  Form the $X$ matrix using the hemisphere and lobe relationships of every pair of ROIs.\;
  Calculate the matrices $Q=I-X(X^TX)^{-1}X^T$, $R=4(Q+4\lambda I)^{-1}$ and $S=4(X^TX)^{-1}X^T$.\;
  Initialize the parameters
  $k=0$, $Z=Z^{(0)}$, $\beta_i=\beta_i^{(0)}$ and evaluate the objective function at the initial values.\;
\While{not stopping condition}{
  Calculate $\left\{\left(\bm{a}_i-\bm{\pi}_i^{(k)}\right)=\bm{a}_i-logit\left(Z^{(k)}+X\beta_i^{(k)} \right)\right\}_{i=1}^N$ and $\mu^{(k)}=\frac{\sum_{i=1}^N\left(\bm{a}_i-\bm{\pi}_i^{(k)}\right)}{N}$\;
  \textbf{Update $\bm{Z}$:} $Z^{(k+1)}=Z^{(k)}+R\left[-\lambda Z^{(k)}+Q\mu^{(k)}\right]$\;
  \textbf{Update $\bm{\beta_i}$:} $\beta_i^{(k+1)}=\beta_i^{(k)}+S\left[\left(\bm{a}_i-\bm{\pi}_i^{(k)}\right)+N\left(Z^{(k+1)}-Z^{(k)}\right)\right]$ $\forall i$\;
  Calculate objective function at $Z^{(k+1)}$ and $\left\{\beta_i^{(k+1)}\right\}_{i=1}^N$\;
  Set k=k+1
 }
\caption{MM Algorithm for estimation of $(Z,\beta_1,\beta_2,\hdots,\beta_N)$}
return $\widehat{Z}=Z^{(k-1)}$ and $\left\{\widehat{\beta}_i=\beta_i^{(k-1)}\right\}_{i=1}^N$
\end{algorithm}

\subsection{Proof of Proposition 1}

The objective function to minimize to estimate the model parameters after excluding the $i^{th}$ unit is given by 
$\mathcal{L}(\widehat{Z},\{\widehat{\beta_j}\}_{j\neq i}) =\frac{1}{N-1}\sum_{j\neq i} \sum_{l=1}^L \left[a_{jl}\widehat{\eta_{jl}}-log(1+e^{\widehat{\eta_{jl}}})\right] -\frac{\lambda}{2}\norm{\widehat{Z}}^2_2$. Using the lemma \ref{invlemma}, and the gradient and hessian given in \ref{gradient_supp} and \ref{hessian_supp} the first Newton Raphson update in $Z$ starting from $\widehat{Z}$ is:
\begin{align}
    Z^{(1)}_{-i} &=\widehat{Z}+\left[\lambda I +\frac{\sum_{j \neq i} W_j}{N-1} -\frac{\sum_{j \neq i} B_{j}^T Q_j^{-1} B_{j}}{N-1} \right]^{-1} \nonumber
    \\ & \qquad \qquad \qquad \qquad \left[\frac{\sum_{j \neq i}\left(I- B_j^T Q_j^{-1} X^T\right) (\bm{a}_j - \widehat{\bm{\pi}}_j)}{(N-1)}-\lambda\widehat{Z}\right]
\end{align}

But since $(\widehat{Z},\widehat{\beta}_1,...,\widehat{\beta}_N)$ is the pMLE, 
$\nabla_{Z} \mathcal{L}(\widehat{Z},\{\widehat{\beta_j}\}_{j=1}^N) = \frac{1}{N }\sum_{i=1}^N (\bm{a}_i-\bm{\widehat{\pi}}_i)-\lambda \widehat{Z} = \bm{0}$ and $\nabla_{\beta_j} \mathcal{L}(\widehat{Z},\{\widehat{\beta_j}\}_{j=1}^N)=\frac{1}{N} X^T (\bm{a}_j-\bm{\widehat{\pi}}_j) = \bm{0}$. So, 
\begin{align}
    \frac{1}{N-1}\sum_{j \neq i}\left(I - B_j^T Q_j^{-1} X^T\right) (\bm{a}_j - \widehat{\bm{\pi}}_j)-\lambda\widehat{Z}&=\frac{1}{N-1}\sum_{j \neq i}(\bm{a}_j - \widehat{\bm{\pi}}_j)-\lambda\widehat{Z} \nonumber \\& =-\frac{1}{N-1}\left[(\bm{a}_i - \widehat{\bm{\pi}}_i)-\lambda\widehat{Z}\right]
\end{align}
So, putting all this together, we get
\begin{align}
    Z^{(1)}_{-i} - \widehat{Z}  &= -\frac{1}{N-1}\left[\lambda I +\frac{\sum_{j \neq i} W_j}{N-1} -\frac{\sum_{j \neq i} B_{j}^T Q_j^{-1} B_{j}}{N-1} \right]^{-1} \left[(\bm{a}_i - \widehat{\bm{\pi}}_i)-\lambda\widehat{Z}\right] \nonumber \\
    &=-\left[\lambda (N-1) I +\sum_{j \neq i} W_j -\sum_{j \neq i} B_{j}^T Q_j^{-1} B_{j} \right]^{-1} \left[(\bm{a}_i - \widehat{\bm{\pi}}_i)-\lambda\widehat{Z}\right]
\end{align}
\qed

\section{Thresholding influence measures for outlier detection}

The basic idea is that most of the subjects have very low values for these influence measures. Only a handful will possibly have high values of these measures. 

\begin{figure}[H]
    \centering
    \includegraphics[width=\linewidth]{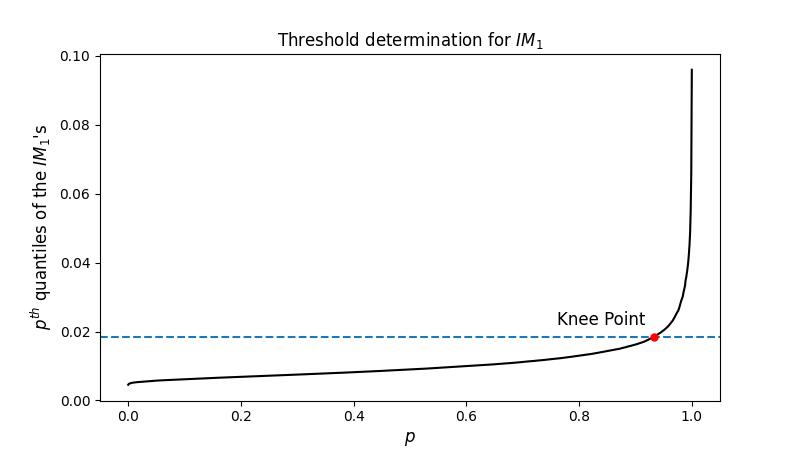}
    \caption{A plot of the quantiles of $IM_2(i)$. Here we plot the $p^{th}$ quantile of the $IM_2(i)$ values for $p \in [0,1]$. As mentioned in section 2 of the paper, the graph remains relatively flat for most of the lower quantiles and then starts to sharply increase as we get to the higher quantiles. The 'elbow' of this plot, i.e. the point where the graph starts increasing sharply can be used as a threshold. This point is shown in red in the graph. It is found using the 'kneedle' algorithm \citep{5961514}. }
    \label{threshhold}
\end{figure}

\section{Additional description of the regions of interest for the circular plots}

The circular plots in the main document, an example of which is given here have several dots arranged around a circular region. Each dot represents a region of interest. These dots are colour-coded to represent which lobe and hemisphere they come from.

\begin{figure}[H]
    \centering
    \includegraphics[width=0.6\linewidth]{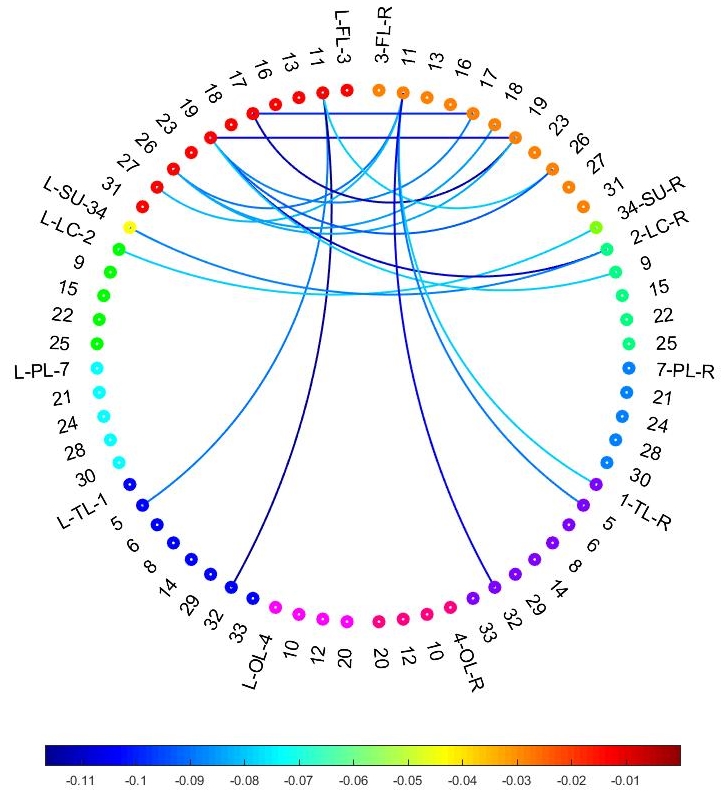}
    \caption{An example of the circular plots in the application section in the main document.}
    \label{cir_plot}
\end{figure}

Each ROI is indicated by a number which is its position in the Desikan atlas. For one of the ROI's in each lobe and hemisphere combination, in addition to the numbers there is also a text in the figure of the form "X-YY". Here "X" is either "L" or "R" representing the left and right hemispheres respectively. "YY" can be FL, SU, LC, PL, TL or OL. These represent the frontal lobe, insula, limbic lobe, parietal lobe, temporal lobe and occipital lobe respectively. All other ROIs with the same colour code belongs to the same lobe and hemisphere as the one with text described here.

\bibliographystyle{unsrtnat}
\bibliography{references.bib}